\documentclass[preprint,byrevtex,onecolumn,12pt,pre,bibnotes]{revtex4}%
\usepackage{amsfonts}
\usepackage{amsmath}
\usepackage{amssymb}
\usepackage{graphicx}%
\providecommand{\U}[1]{\protect\rule{.1in}{.1in}}
\newtheorem{theorem}{Theorem}

\newenvironment{proof}[1][Proof]{\noindent\textbf{#1.} }{\ \rule{0.5em}{0.5em}}
\begin{document}
\preprint{UATP/1201}
\title{Nonequilibrium Thermodynamics. Symmetric and Unique Formulation of the First
Law, Statistical Definition of Heat and Work, Adiabatic Theorem and the Fate
of the Clausius Inequality: A \ Microscopic View}
\author{P.D. Gujrati}
\email{pdg@uakron.edu}
\affiliation{Department of Physics, Department of Polymer Science, The University of Akron,
Akron, OH 44325}

\begin{abstract}
The status of heat and work in nonequilibrium thermodynamics is quite
confusing and non-unique at present with conflicting interpretations even
after a long history of the first law $dE(t)=$ $d_{\text{e}}Q(t)-dW_{\text{e}%
}(t)$ in terms of exchange heat and work, and is far from settled. Moreover,
the exchange quantities lack certain symmetry (see text). By generalizing the
\emph{traditional} concept to also include their time-dependent irreversible
components $d_{\text{i}}Q(t)$ and $d_{\text{i}}W(t)$ allows us to express the
first law in a symmetric form $dE(t)=$ $dQ(t)-dW(t)$ in which $dQ(t)$ and work
$dW(t)$ appear on equal footing and possess the symmetry.\emph{ }We prove that
$d_{\text{i}}Q(t)\equiv d_{\text{i}}W(t)$; as a consequence, irreversible work
turns into irreversible heat. Statistical analysis in terms of microstate
probabilities $p_{i}(t)$ \emph{uniquely} identifies $dW(t)$ as
\emph{isentropic} and $dQ(t)$ as \emph{isometric} (see text) change in
$dE(t)$, a result known in equilibrium. We show that such a clear separation
does not occur for $d_{\text{e}}Q(t)$ and $dW_{\text{e}}(t)$. Hence, our new
formulation of the first law provides tremendous advantages and results in an
extremely useful formulation of non-equilibrium thermodynamics, as we have
shown recently [Phys. Rev. E \textbf{81}, 051130 (2010); \textit{ibid}
\textbf{85}, 041128 and 041129 (2012)]. We prove that an adiabatic process
does not alter $p_{i}$. \emph{All these results remain valid no matter how far
the system is out of equilibrium. }When the system is in internal equilibrium,
$dQ(t)\equiv T(t)dS(t)$ in terms of the instantaneous temperature $T(t)$ of
the system, which is reminiscent of equilibrium, even though, neither
$d_{\text{e}}Q(t)\equiv T(t)d_{\text{e}}S(t)$ nor $d_{\text{i}}Q(t)\equiv
T(t)d_{\text{i}}S(t)$. Indeed, $d_{\text{i}}Q(t)\ $and $d_{\text{i}}S(t)$ have
very different physics. We express these quantities in terms of $d_{\text{e}%
}p_{i}(t)$ and $d_{\text{i}}p_{i}(t)$, and demonstrate that $p_{i}(t)$ has a
form very different from that in equilibrium.\ \emph{The first and second laws
are no longer independent so that we need only one law, which is again
reminiscent of equilibrium. }The traditional formulas like the Clausius
inequality $\oint d_{\text{e}}Q(t)/T_{0}<0$, $\Delta_{\text{e}}W<-\Delta
\left[  E(t-T_{0}S(t))\right]  $, etc. become equalities $\oint
dQ(t)/T(t)\equiv0,$ $\Delta W=-\Delta\left[  E(t-T(t)S(t)\right]  $, etc, a
quite remarkable but unexpected result in view of the fact that $\Delta
_{\text{i}}S(t)>0.$ We identify the uncompensated transformation $N(t,\tau)$
during a cycle. We determine the irreversible components in two simple cases
to show the usefulness of our approach; here, the traditional formulation is
of no use. Our extension bring about a very strong parallel between
equilibrium and non-equilibrium thermodynamics, except that one has
irreversible entropy generation $d_{\text{i}}S(t)>0$ in the latter.

\end{abstract}
\date{\today}
\maketitle

\section{Introduction and Controversy}

\subsection{Heat and Work}

Gislason and Craig \cite{Gislason} recently remarked that the definition of
work in nonequilibrium "...thermodynamics processes remains a contentious
topic," a rather surprising statement, as the field of thermodynamics is an
old discipline. However, there is some truth to their critique, which was
motivated by an earlier paper by Bertrand \cite{Bertrand}, who revisited the
confusion first noted by Bauman \cite{Bauman} about different formulation of
work [the \emph{exchange heat work }$d_{\text{e}}W(t)=P_{0}dV(t)$\ or the
\emph{generalized work}$\ dW(t)=P(t)dV(t)$, see Refs.
\cite{Prigogine,deGroot,note-0,Gujrati-I,Gujrati-II,Gujrati-Heat-Work,Gujrati-III}
and below] in terms of the external ($P_{0}$) and instantaneous or internal
($P(t)$) pressures \cite{note-1} at a given instant $t$, see Fig.
\ref{Fig_System}, and discussed by many others
\cite{Alicki,Deffner,Kestin,Kievelson,Gislason-TwoSystems,Bizarro,Anacleto-SecondLaw,Anacleto-DissipativeWork,Honig,Jarzynski,Jarzynski-1,Rubi,Rubi-1,Peliti,Peliti-1,Cohen,Sekimoto,Seifert,Nieuwenhuizen,Crooks,Jarzynski0,Ritort}%
\ since then with no consensus \cite[p.181, Vol. 1]{Kestin}. Traditional
formulation of nonequilibrium statistical mechanics and thermodynamics
\cite{Landau,Gibbs} follows a mechanistic approach in which the system of
interest follows its (classical or quantum) mechanical evolution in time.
Being a mechanical concept, work is easier to identify in thermodynamics than
heat; the latter is not possible to be identified with any mechanical force.
Therefore, one usually identifies heat by first identifying work and then
subtracting the latter from the energy change \cite{Landau} in the first law.
This means that different formulations of work will result in different heats.
We refer the reader to the above references for an interesting history of the
confusion. Gislason and Craig \cite{Gislason-TwoSystems} list twenty-six
representative textbooks including \cite{Kirkwood,Zemansky} where the
pressure-volume work and heat are defined so differently that they are not
equivalent in the presence of irreversibility, and there appears to be no
consensus about their right formulation so far; see the recent debate in the
field \cite{Jarzynski,Jarzynski-1,Rubi,Rubi-1,Peliti,Peliti-1,Cohen}.
Unfortunately, none of the latter sources consider \emph{internal variables}
\cite[and references therein]{Gujrati-I,Gujrati-II,Maugin} that are needed to
describe irreversible processes but do not appear explicitly in the
Hamiltonian. They also do not consider \emph{thermodynamic forces}
\cite{Gujrati-I,Gujrati-II,Prigogine} that are non-zero when the system is
away from equilibrium.

The attempt to define heat by first defining work leaves the concept of heat
devoid of clear physical significance, especially since it depends on what we
mean by work. In many cases, work is identified by considering the work
performed by bodies external to the system (the medium; see Fig.
\ref{Fig_System}), which may have nothing to do with the work done by the
system, especially when we consider irreversible processes. This has created a
lot of confusion in the literature when dealing with irreversible processes;
see Refs. \cite{Jarzynski,Jarzynski-1,Rubi,Rubi-1,Peliti,Peliti-1,Cohen} for a
recent discussion; see also \cite{Gujrati-Heat-Work}. The traditional
formulation of statistical mechanics and thermodynamics \cite{Landau,Gibbs}
follows a mechanistic approach in which the system of interest, to be precise
its microstate, follows its (classical or quantum) mechanical evolution in
time. The only difference between mechanics and thermodynamics is that the
evolution in thermodynamics is \emph{always} stochastic, which makes the
evolution irreversible in accordance with the second law; see for example a
recent review \cite{Gujrati-Symmetry}. Therefore, it is crucial to consider a
\emph{statistical foundation} of nonequilibrium thermodynamics for a better
understanding of heat and work, a project we have also initiated recently
\cite{Gujrati-Symmetry,Gujrati-I,Gujrati-II,Gujrati-III,Gujrati-IV}. The
concept of work follows in a trivial manner by taking (a \emph{stochastic
average} of) its mechanical analog (work done by generalized forces in the
Hamiltonian formulation) \cite{Landau}. and was shown to be given by
$dW_{\text{V}}(t)=P(t)dV(t)$ for the pressure work; see Ref. \cite{Gujrati-I}.
These forces can be controlled by an observer and determine the
\emph{observables} in the system; see for example \cite{Gujrati-I,Gujrati-II}.
It was shown in Ref. \cite{Gujrati-I,Gujrati-II} that the first law can
\emph{also} be written in terms of the \emph{generalized heat} $dQ(t)$
\emph{added to} and the \emph{generalized pressure work} $dW(t)$ \emph{done
by} the system:
\begin{equation}
dE(t)=dQ(t)-dW(t), \label{First_Law}%
\end{equation}
where%
\begin{equation}
d_{\text{i}}Q(t)\equiv dQ(t)-d_{\text{e}}Q(t),d_{\text{i}}W(t)\equiv
dW(t)-d_{\text{e}}W(t), \label{Generalized-Heat-Work}%
\end{equation}
denote irriversible heat and work, respectively, generated within the system
\cite{note-0}. However, the interpretation of the above new formulation, its
statistical basis, and its realtionship with the traditional formulation
$dE(t)=d_{\text{e}}Q(t)-d_{\text{e}}W(t)$ in terms of exchange heat
$d_{\text{e}}Q(t)=T_{0}d_{\text{e}}S(t)$ and work $d_{\text{e}}W(t)=P_{0}%
dV(t)$ was not explored in earlier publications
\cite{Gujrati-I,Gujrati-II,Gujrati-Heat-Work,Gujrati-III}.

For an isolated system, the mechanical observables, collectively denoted by
the set $\mathbf{X}$, are additive integrals of motion
\cite{Landau-Mech,Landau} such as the energy $E$, the number of particles
$N_{j}$ of different species, and linear and angular momenta $\mathbf{P}$\ and
$\mathbf{M}$ of the system, respectively, whose values depend on the
preparation of the system. When the motion is confined to a finite region of
space, the volume $V$ of this region also characterizes the system as a
constraint. The \emph{constraints} are also treated as observables.The
internal variables unfortunately cannot be controlled by an observer. Despite
this, the concept of work needs to incorporate the additional work required in
changing the internal variables. This gives rise to some complication in
identifying the form of work associated with them. The collection of
observables and internal variables will be called \emph{state variables} and
will be denoted by the set $\mathbf{Z}$. The macrostate can be characterized
in a larger state space of $\mathbf{Z}$, even though the observed macrostate
is characterized by $\mathbf{X}$.\ The internal variables do not remain
constant for an isolated system as the system relaxes. In this sense, we must
treat the isolated system as interacting with an appropriate "fictitious"
medium to allow for the variation of internal variables, a fact not emphasized
in the literature to the best of our knowledge; see Ref. \cite{Gujrati-II}. As
both energy and work are mechanical concepts, it is trivial to identify their
statistical averages over microstates whether the system is in or out of equilibrium.

\subsection{Controversy}

As the concept of heat has no analog in mechanics, it is usually identified
indirectly as that energy change which is \emph{not} work. \emph{This kind of
approach makes heat and work not unique}: a change in one affects the other;
only their difference has a well-defined meaning. This has created a lot of
confusion in the field. Our goal in this work is to establish that there
exists a consistent procedure in which the partition is \emph{unique} so that
work and heat are no longer arbitrary in irreversible processes. They are
given by generaized work and heat despite the controversy. The first law of
thermodynamics, see for example, Refs. \cite{Prigogine,Landau}, is a very
general statement and is supposed to be valid for all processes, and not just
equilibrium processes. The traditional formulation of the first law in terms
of the \emph{exchange heat and work} is not only oblivious to the internal
variables and thermodynamic forces, whose presence must control the approach
to equilibrium, but most significantly, is also oblivious to the violation of
the second law; see below. This is why we need to use both the first and the
second laws in nonequilibrium thermodynamics. How can a formulation of a
fundamental law of physics allow for the violation of another fundamental law
of physics? There is also an asymmetry between heat and work that will be
elaborated below in that the two terms are not on an equal footing. There is
obviously no problem for reversible processes.

We mostly consider mechanical work including dissipation, but the arguments
are valid for all kinds of work; see Sec. \ref{Sect_Other_State_Variables},
however. Zemansky \cite[p.73]{Zemansky} defines heat as energy exchange "...by
virtue of a temperature difference only." Unfortunately, this rules out any
isothermal heat exchange and cannot be considered general. Kirkwood and
Oppenheim define heat as energy exchange resulting in "...the temperature
increment..." (which rules out phase changes requiring latent heat) and later
note that the work may be converted to heat due to frictional dissipation
\cite[pp. 16,17]{Kirkwood} as was first observed by Count Rumford in 1798
during the boring of cannon \cite{Prigogine}. These are two of the examples
where heat is defined directly without first identifying work; they have
limitations and cannot be considered general, especially when irreversibility
is present. The situation with work is just as confusing as we noted above
with its two different formulations: $d_{\text{e}}W(t)=P_{0}dV(t)$\ or
$dW(t)=P(t)dV(t)$.

Dissipation always gives rise to positive entropy generation due to
irreversibility and also raises the temperature such as due to friction or the
Joule heat in resistors. Therefore, it is natural to account for such viscous
dissipation in work and heat when dealing with nonequilibrium systems, as they
are integral to the system and dictate its relaxation. The fact that
literature is not very clear on how to incorporate viscous dissipation has
motivated this work; see however \cite{Bizarro,Anacleto-DissipativeWork}, but
the authors do not take the discussion far enough to obtain the results
derived here. Recently, dissipative forces are explicitly considered in
stochastic trajectory thermodynamics
\cite{Jarzynski,Jarzynski-1,Rubi,Rubi-1,Peliti,Peliti-1,Cohen,Sekimoto,Seifert,Crooks,Jarzynski0,Ritort}%
. However, the approach differs from our approach \cite{note-2} in important
ways and has also given rise to controversy, and remains contentious
\cite{Jarzynski,Jarzynski-1,Rubi,Rubi-1,Peliti,Peliti-1,Cohen}. In addition,
the approach is not general as it is limited to isothermal variations and to
cases where Langevin dynamics is applicable; see, however, Hoover and Hoover
\cite{Hoover} where an example of a time-reversible deterministic Hamiltonian
system is given.

\subsection{New Results}

We interpret heat and work used in the recent reformulation in Eq.
(\ref{First_Law}) that was proposed in Refs.
\cite{Gujrati-I,Gujrati-II,Gujrati-III} and follow the consequences with an
aim to develop their statistical definition. The formulation makes the first
law identical to the second law (in the guise of the Gibbs fundamental
relation) In general, $d_{\text{e}}Q(t)\equiv T_{0}d_{\text{e}}S(t),$ but
$dS(t)$\ and $dQ(t)$\ are not related to each other, except in internal
equilibrium \cite{Gujrati-I}. Moreover, it would be incorrect to conclude from
$dS(t)\equiv d_{\text{e}}Q(t)/T_{0}+d_{\text{i}}S(t)$ that $d_{\text{i}%
}Q(t)=T_{0}d_{\text{i}}S(t)$; \ see later. The aim of any theory of
nonequilibrium thermodynamics is to determine the entropy change $dS(t)$.
Therefore, the determination of $d_{\text{i}}S(t)$ becomes the focus in any
investigation of a body in the traditional formulation. The second law is
reflected in the inequality $d_{\text{i}}S(t)>0$ for any irreversible process.
The statistical analysis provides an elegant formulation of nonequilibrium
thermodynamics in which there is a \emph{unique} and \emph{natural}
distinction between nonequilibrium work and heat in that the generalized work
represents \emph{isentropic change in the (internal) energy} and the
generalized heat represents the \emph{isometric} (constant extensive state
variables excluding the energy) change in the energy. The latter change\emph{
}results purely from the entropy change. The unique partition of energy
remains valid no matter how far the system is out of equilibrium. The
assumption of \emph{internal equilibrium} allows us to express the first law
in terms of the \emph{instantaneous} (or internal) fields. This brings about a
very close parallel between nonequilibrium and equilibrium processes such as
$dS(t)\equiv dQ(t)/T(t)$\ and the existence of a theorem for irreversible
processes identical in spirit to the \emph{adiabtic theorem} \cite{Landau} for
equilibrium processes. \ The approach proposed recently in Refs.
\cite{Gujrati-Symmetry,Gujrati-I,Gujrati-II,Gujrati-III} deals directly with
$dS(t)$ without a need to use $d_{\text{e}}S(t)$ and $d_{\text{i}}S(t)$
separately, although they can also be evaluated in our approach. Therefore,
our approach should be quite useful as the entropy is a state variable; see
below, however, for other advantages. A particular symmetry is explicitly seen
in our formulation in that both $dQ$ and $dW$\ not only do not change with the
nature of the process but also exhibit an identical formulation in terms of
entropy and volume, respectively. In other words, they are found to be on an
equal footing.\ As a consequence, the first law becomes identical to the
second law (as the Gibbs fundamental relation) so we only deal with a single
law. The Clausius inequality \cite{Clausius} turns into an equality in
\emph{all} cases as $dQ(t)/T(t)$ becomes a state variable, the work is
expresssed as an equality, as if we are dealing with equilibrium processes, a
quite remarkable result in its own right, even though there is irreversible
entropy generation. It has been recently suggested \cite{Anacleto-SecondLaw}
that the use of internal fields is not always consistent with the second law.
We find no such problem in our approach.
%TCIMACRO{\FRAME{ftbpFU}{3.4584in}{1.772in}{0pt}{\Qcb{Schematic representation
%of $\Sigma$, $\widetilde{\Sigma}$ and $\Sigma_{0}$. We assume that $\Sigma$
%and $\widetilde{\Sigma}$ are homogeneous and in internal equilibrium, but not
%in equilibrium with each other. The internal fields $T(t),P(t),\cdots$ fof
%$\Sigma$ and $T_{0},P_{0},\cdots$ of $\widetilde{\Sigma}$ are not the same
%unless they are in equilibrium with each other. There will be viscous
%dissipation in $\Sigma$ when not in equilibrium with $\widetilde{\Sigma}$.}%
%}{\Qlb{Fig_System}}{system_modified_1.eps}%
%width 3.4584in;  height 1.772in;  depth 0pt;  original-width 5.1673in;
%original-height 2.4915in;  cropleft "0.0775";  croptop "1";
%cropright "1.0258";  cropbottom "0";
%filename 'System_Modified_1.eps';file-properties "XNPEU";}}}%
%BeginExpansion
\begin{figure}
[ptb]
\begin{center}
\includegraphics[
trim=0.400466in 0.000000in -0.133316in 0.000000in,
height=1.772in,
width=3.4584in
]%
{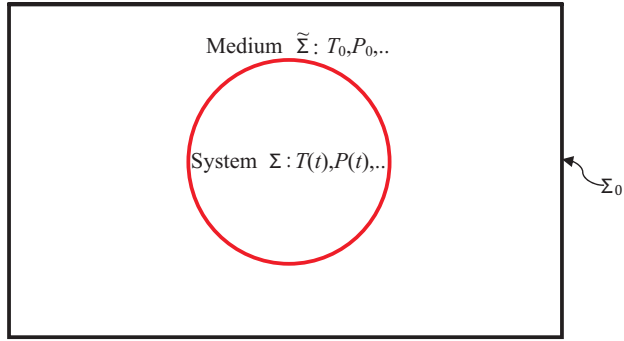}%
\caption{Schematic representation of $\Sigma$, $\widetilde{\Sigma}$ and
$\Sigma_{0}$. We assume that $\Sigma$ and $\widetilde{\Sigma}$ are homogeneous
and in internal equilibrium, but not in equilibrium with each other. The
internal fields $T(t),P(t),\cdots$ fof $\Sigma$ and $T_{0},P_{0},\cdots$ of
$\widetilde{\Sigma}$ are not the same unless they are in equilibrium with each
other. There will be viscous dissipation in $\Sigma$ when not in equilibrium
with $\widetilde{\Sigma}$.}%
\label{Fig_System}%
\end{center}
\end{figure}
%EndExpansion

The new formulation of the first law contains not only the observables and the
instantaneous fields of the system but also explicitly contains internal
variables so no information about the system is lost. This makes the new
formulation quite beneficial. The instantaneous temperature, pressure, etc. of
the system need not be identical to those of the medium when the system is out
of equilibrium with the medium as has become clear recently
\cite{Nieuwenhuizen,Langer,Gujrati-I}. Thus, our approach also overcomes the
objection raised by Cohen and Mauzerall \cite{Cohen} against the stochastic
trajectory thermodynamics \cite{Sekimoto,Jarzynski,Seifert} to which our
approach is easily extended. We will establish that there are several benefits
in accounting for viscous dissipation within the system in the first law, but
care must be exercised. Our approach, which \emph{assumes} the existence of
entropy even for nonequilibrium states \cite{Gujrati-Symmetry} as was first
proposed by Clausius \cite{Clausius}, is otherwise very general as we do not
restrict ourselves to any particular dynamics and to only isothermal
variations. It allows us to incorporate viscous dissipation, due to the
presence of thermodynamic forces and internal variables, in the discussion
explicitly within this general framework. We consider our system $\Sigma$ (see
Fig. \ref{Fig_System}) surrounded by a very large medium $\widetilde{\Sigma}$
so large that its fields such as its temperature $T_{0}$, pressure $P_{0}$,
etc.\ are not affected by the system. They form an isolated system $\Sigma
_{0}$.\ We consider all systems to be stationary; their relative motion will
be considered only in Sec. \ref{Sect_Applications}. The surface separating
$\Sigma$ and $\widetilde{\Sigma}$, which may represent a piston, will be
treated as having no interesting thermodynamics of its own although it may
participate in irreversibility due to field (such as temperature, pressure,
etc.) differences across it. Because of the enormous size of the medium with
respect to that of the system, all irreversible components in $\Sigma_{0}%
$\ appear within $\Sigma$ ($d_{\text{i}}Q(t)\equiv d_{\text{i}}Q_{0}%
(t),d_{\text{i}}W(t)\equiv d_{\text{i}}W_{0}(t)$) although this is not always
so in the literature \cite{note-3}. The situation of finite surroundings has
been considered by Bizarro \cite{Bizarro} and by us in Ref. \cite{Gujrati-II}%
.) In the following, all extensive quantities pertaining to $\widetilde
{\Sigma}$ and $\Sigma_{0}$ carry an annotation tilde or a suffix $0$,
respectively, and those pertaining to $\Sigma$ carry no suffix. We will use
\emph{body} to refer to any\ one of the above three systems and use symbols
without any suffix to denote its quantities.

Our discussion is easily extended to the case when $\widetilde{\Sigma}$ is
comparable to $\Sigma$ in size, as is easily seen in Secs.
\ref{Sect_Applications} and \ref{Sec-Closed-System}, and Ref.
\cite{Gujrati-II}. In this case, we will refer to $\widetilde{\Sigma}$ as the
surroundings, whose fields will also change with time. However, to keep the
discussion simple, we will treat $\widetilde{\Sigma}$ as an extensively large
medium for the most part.

The layout of the paper is as follows. We briefly review the traditional
formulation of the first law, its lack of symmetry, confusion and its
limitations in the next section. In Sec. \ref{Sec_General_Consideration}, we
follow our newly developed nonequilibrium thermodynamics and unravel the
significance of heat and work in that approach. Several important results are
derived there in the forms of theorems. Some of the results were announced
earlier in Ref. \cite{Gujrati-Heat-Work} but have now been expanded. In Sec.
\ref{Sec_Work_Second_Law}, we argue that only the generalized work is
consistent with the second law. The statistical definition of heat and work is
taken up in Sec. \ref{Sec_Stat_Concepts}, which forms the core of the present
work, where we show that this definition is identical with heat and work
discussed in Sec. \ref{Sec_General_Consideration} and used in the earlier work
\cite{Gujrati-I,Gujrati-II,Gujrati-III}. This section also contains many
important results including the adiabatic theorem for irreversible processes.
A general expression for microstate probabilities is derived here, which
clearly shows nonequilibrium effects in their formulation. We partition
microstate probability changes into external and internal parts $d_{\text{e}%
}p_{i}(t),d_{\text{i}}p_{i}(t)$ to obtain the statistical formulations
for\ $d_{\text{e}}Q(t),d_{\text{i}}Q(t)$ etc. We use the example of an ideal
quantum gas to show how $d_{\text{e}}p_{i}(t),d_{\text{i}}p_{i}(t)$ can be
computed. In the following section, we discuss the Clausius inequalities
(there are two different ones) and the Clausius equality and the work
equality. In Sec. \ref{Sect_Applications}, we consider two applications of our
formulation, where the traditional formulation cannot be applied. In Sec.
\ref{Sect_Other_State_Variables}, we extend our discussion to include an
additional observable. The results derived here are used in Sec.
\ref{Sec-Closed-System} to study a closed system which is allowed to exchange
some kind of "work" with a thermally isolated external object. This is a
classic prototype model studied extensively; see for example, Ref.
\cite{Landau}. The last section contains a brief summary of results and a list
of benefits of our approach.

\section{Traditional Formulation of the First
Law\label{sect_Traditional Formulation}}

\subsection{Traditional Formulation}

To truly appreciate our contribution, it is useful to consider how the first
law is traditionally expressed. We will only consider a single internal
variable $\xi$ for simplicity. Similarly, we will consider only $E$ and $V$
for simplicity as observables with the number of particles $N$ (only a single
species) held fixed, unless noted otherwise. Other variables are easy to
include in the approach as we do in Sec. \ref{Sect_Other_State_Variables}.
Traditionally, $dQ(t)$ represents the amount of heat exchange $d_{\text{e}%
}Q(t)$, so that $-dW(t)\ $is identified with the work exchange $-d_{\text{e}%
}W(t)\equiv d\widetilde{W}(t)=P_{0}d\widetilde{V}(t)$ by the medium to the
system, giving $d_{\text{e}}W(t)=P_{0}dV(t)$. This is true even if the system
cannot be assigned any pressure or if its\emph{ instantaneous} pressure $P(t)$
is different from $P_{0}$. As the net heat exchange $d_{\text{e}%
}Q(t)+d_{\text{e}}\widetilde{Q}(t)=0$, we immediately verify $d_{\text{e}%
}Q(t)\equiv T_{0}d_{\text{e}}S(t).$ There is no such general relation
relations for other heats: $dQ(t)\neq T_{0}dS(t)$ and $d_{\text{i}}Q(t)\neq
T_{0}d_{\text{i}}S(t).$ The traditional formulation of the first law for a
general process reads
\begin{equation}
dE(t)\equiv d_{\text{e}}Q(t)-d_{\text{e}}W(t)\equiv T_{0}d_{\text{e}%
}S(t)-P_{0}dV(t) \label{Standard_Heat_Sum}%
\end{equation}
expressed in terms of either exchange quantities or external fields of
$\widetilde{\Sigma}$. Only when the process is reversible that we have
$dE(t)\equiv T_{0}dS(t)-P_{0}dV(t).$ The external fields are conjugates to the
observables in $\widetilde{\mathbf{X}}$ ($\widetilde{E},\widetilde{V}$), with
the medium affinity $A_{0}$ conjugate to $\widetilde{\xi}$ vanishing.\ Thus,
the above formulation of the first law is, as said earlier, oblivious to the
internal variables and will not be considered when using this formulation. The
following inequality for a cycle, commonly known as the \emph{Clausius
inequality}, follows from $dS>d_{\text{e}}Q(t)/T_{0}$,%
\begin{equation}%
%TCIMACRO{\toint }%
%BeginExpansion
{\textstyle\oint}
%EndExpansion
d_{\text{e}}S(t)\equiv%
%TCIMACRO{\toint }%
%BeginExpansion
{\textstyle\oint}
%EndExpansion
d_{\text{e}}Q(t)/T_{0}<0, \label{Clausius_Inequality_00}%
\end{equation}

\subsection{Confusion about Work and Heat}

The situation regarding $dW(t)$ is not always clear. Kondepudi and Prigogine
use $dW(t)=PdV(t)$, where $P$ "...is the pressure at the moving surface," but
they do not mention whether the form is applicable to all processes. Landau
and Lifshitz are explicit and state that $dW(t)=P(t)dV(t)$ for reversible and
irreversible processes \cite[p.45]{Landau}. They require for this the
\emph{existence} of mechanical equilibrium (and so do Refs.
\cite{Kirkwood,Zemansky}) within $\Sigma$ so that at each instant during the
process $P(t)$ must be uniform throughout the body; its equality with $P_{0}$
is not required. However, they do not discuss $dQ(t)$ when they consider
$\Sigma$ out of equilibrium with $\widetilde{\Sigma}$ \cite[Sect. 20]{Landau}.
If we use $dW(t)=P(t)dV(t)$ for the work \emph{done by} $\Sigma$, then this
will alter the heat $dQ(t)$ \emph{added to} $\Sigma$ in Eq. (\ref{First_Law}).
This follows immediately from the fact, not appreciated in the literature to
the best of our knowledge, that $dE(t)$ must be \emph{invariant} to the choice
of internal or external fields. Also, to the best of our knowledge, the issue
of the actual forms of $dQ(t)$ and what is the correct form of $dW(t)$ for
nonequilibrium processes has not been settled in the literature. Indeed,
Kestin \cite[Sect. 5.12]{Kestin} clearly states that distinguishing heat and
work in nonequilibrium states is not unambiguous. We will later argue
otherwise in this work.

\subsection{Lack of Symmetry and Disconnection with the Second Law}

The exchange heat and work $d_{\text{e}}Q(t)$ and $d_{\text{e}}W(t)$\ are very
different\ for irreversible processes, since $d_{\text{e}}Q(t)\equiv
T_{0}d_{\text{e}}S(t)$ is in terms of $d_{\text{e}}S(t)$, which is not the
change in a state variable, while $d_{\text{e}}W(t)=P_{0}dV(t)$ depends on the
change in a state variable.\ Thus, there is an \emph{asymmetry} between the
two in the second equation in Eq. (\ref{Standard_Heat_Sum}) in that they are
not on an equal footing. However, the main disadvantage of the formulation is
that it is always valid, even for a process that violates the second law by
having the exchange heat flow from a colder to a hotter object. This is why we
need both laws in the traditional formulation of non-equilibrium
thermodynamics. We believe that an elegant formulation of a fundamental law
like the first law should not only satisfy the other fundamental law, the
second law but also exhibit as much symmetry as possible.

\subsection{Limitations\label{Subsection_Limitations}}

As the Gibbs fundamental relation, see Eq. (\ref{Gibbs_Fundamental_Relations}%
), for the system explicitly contains the internal variable, the first law in
its traditional formulation and the Gibbs fundamental relation, which codifies
the second law, have different contents for irreversible processes; in
particular, the latter contains more information than the former. Only for
reversible processes, for which internal variables are no longer independent
of the observables, the first law can be used to determine the change $\Delta
S(t)=[dE(t)+P_{0}dV(t)]/T_{0}$ in the entropy of a body. This is not true when
we deal with irreversible processes. This limits the usefulness of the first
law. We now list some of the important limitations of the traditional
formulation below.

\begin{enumerate}
\item[(1)] The law is oblivious to the violation of the second law.

\item[(2)] There is the above mentioned asymmetry between heat and work.

\item[(3)] As both $d_{\text{e}}Q(t)$ and $d_{\text{e}}W(t)$ are determined by
the medium, their knowledge does not provide us with any \emph{direct}
information about the system or its entropy change $\Delta S(t)$ in an
irreversible process. This is easily seen by considering an isolated body. As
$d_{\text{e}}Q(t)=d_{\text{e}}W(t)=0$, the first law has no useful content.
Moreover, it cannot be used to determine $dS(t)=d_{\text{i}}S(t)$ as the body
changes from some state A$^{^{\prime}}$ to a nearby state A unless both states
are equilibrium states. In the latter case, one can compute $dS(t)$ by
consider some equilibrium path connecting the two states. This approach will
not work if one of the two states is or both are out of equilibrium.

\item[(4)] Work and heat cannot always be unambiguously distinguished, a point
already made very strongly by Kestin \cite[Sect. 5.12]{Kestin} and which is at
the heart of the dispute discussed above.

\item[(5)] The heat and work do not always flow through the boundary
\cite[footnote on p. 176]{Kestin}.

\item[(6)] As $d_{\text{e}}S(t)$ can be determined from $d_{\text{e}}Q(t)$,
the aim of any nonequilibrium thermodynamic investigation using the
traditional formulation is to determine the \emph{irreversible entropy change}
$d_{\text{i}}S(t)$. For this, one needs to invoke the Gibbs fundamental
relation in addition to the traditional first law; see for example de Groot
and Mazur \cite{deGroot}.

\item[(7)] As internal variables $\boldsymbol{\xi}(t)$
\cite{deGroot,Langer,Gujrati-I,Gujrati-II,Gujrati-III,Maugin,Bridgman}, which
are very common in nonequilibrium systems such as glasses or in chemical
reactions, play an important role in nonequilibrium thermodynamics, their
behavior will strongly affect the dissipation within the system.
Unfortunately, these variables do \emph{not} couple to the medium
\cite{Maugin}; hence, they do not appear in $d_{\text{e}}Q(t)$ and
$d_{\text{e}}W(t)$, although they control the thermodynamic relaxation and, in
particular, the Gibbs fundamental relation for the system.
\end{enumerate}

\section{General Consideration\label{Sec_General_Consideration}}

As a generalization of the equilibrium concept, we use the \emph{instantaneous
values} of the state variables containing the observables ($\mathbf{X(}%
t\mathbf{):}E(t),V(t)$) and internal variables ($\xi(t)$)
\cite{deGroot,Gujrati-I,Gujrati-II,Gujrati-III} to identify the state of the
body. A body can be in the same state at different times. This is important so
that a system can go through a cyclic process in which the system comes back
to the same initial state at a later instant. However, the entropy of a body
at some instant, besides being a function of the state variables, may also
have an \emph{explicit} dependence on time: $S(t,E(t),V(t),\xi(t)).$ Thus, in
general, the entropy will not be a state function. The first and second laws
are not useful for any computation unless we can ascribe temperatures,
pressures, etc. to $\Sigma$. This requires $\Sigma$ and $\widetilde{\Sigma}$
to be in \emph{internal equilibrium} \cite{Landau,Gujrati-I,Gujrati-II} when
their instantaneous entropies become \emph{state functions }$S(t)=S\left[
E(t),V(t),\xi(t)\right]  ,\widetilde{S}(t)=\widetilde{S}\left[  \widetilde
{E}(t),\widetilde{V}(t),\widetilde{\xi}(t)\right]  $ of (time-dependent) state
variables. Let $W(t)\equiv W\left[  E(t),V(t),\xi(t)\right]  $ and
$\widetilde{W}(t)\equiv\widetilde{W}\left[  \widetilde{E}(t),\widetilde
{V}(t),\widetilde{\boldsymbol{\xi}}(t)\right]  $\ denote the number of
microstates consistent with the state variables for each of them. Then, as
discussed in Refs. \cite{Gujrati-I,Gujrati-II}, we have
\begin{subequations}
\label{Internal_Eq_S_S_tilde}%
\begin{equation}
S(t)=\ln W(t),\widetilde{S}(t)=\ln\widetilde{W}(t) \label{Internall_Eq_S}%
\end{equation}
in other words, the microstates in $W(t)$ or in $\widetilde{W}(t)$ are
\emph{equally probable}.

The temperatures, pressures and affinities (we introduce $\beta(t)\equiv
1/T(t)$ and $\beta_{0}\equiv1/T_{0}$) are given by appropriate standard
derivatives of the entropies:%
\end{subequations}
\begin{subequations}
\begin{align}
\beta(t)  &  =\partial S(t)/\partial E(t),\beta(t)P(t)=\partial S(t)/\partial
V(t),\beta(t)A(t)=(\partial S(t)/\partial\xi
(t);\label{Thermodynamic-Fields-System}\\
\beta_{0}  &  =\partial\widetilde{S}(t)/\partial\widetilde{E}(t),\beta
_{0}P_{0}=\partial\widetilde{S}(t)/\partial\widetilde{V}(t),\partial
\widetilde{S}(t)/\partial\widetilde{V}(t)=0.
\label{Thermodynamic-Fields-Medium}%
\end{align}
The Gibbs fundamental relations are given by
\cite{Landau,Gujrati-I,Gujrati-II}%
\end{subequations}
\begin{equation}
dE(t)=T(t)dS(t)-P(t)dV(t)-A(t)d\xi(t),\ \ d\widetilde{E}(t)=T_{0}%
d\widetilde{S}(t)-P_{0}d\widetilde{V}(t); \label{Gibbs_Fundamental_Relations}%
\end{equation}
The validity of Eq. (\ref{Gibbs_Fundamental_Relations}) requires $\Sigma$ and
$\widetilde{\Sigma}$ to be independently homogeneous
\cite{Gujrati-I,Gujrati-II} and in internal equilibrium. We now prove the
following trivial but important theorem.

\begin{theorem}
\label{Theorem_Irreversible-Heat-Work}Irreversible work and irreversible heat
have identical values:%
\begin{equation}
d_{\text{i}}Q(t)\equiv d_{\text{i}}W(t),
\label{Irreversible_Heat_Work_equality}%
\end{equation}

\end{theorem}

\begin{proof}
We have\ $dE_{0}=dV_{0}=A_{0}=0$ for $\Sigma_{0}$. The application of the
first law for $\Sigma_{0}$\ using generalized heat and work, see Eq.
(\ref{Generalized-Heat-Work}), gives%
\begin{equation}
dE_{0}=dQ_{0}(t)-dW_{0}(t)\equiv d_{\text{i}}Q_{0}(t)-d_{\text{i}}W_{0}(t)=0;
\label{Gibbs_Fundamental_Relation_Isolated}%
\end{equation}
there is no exchange heat and no exchange work for $\Sigma_{0}$. As the
irreversibility is only associated with the system $\Sigma$, we have
$d_{\text{i}}Q(t)=d_{\text{i}}Q_{0}(t)$ and $d_{\text{i}}W(t)=d_{\text{i}%
}W_{0}(t)$. The desired equality in Eq. (\ref{Irreversible_Heat_Work_equality}%
) now follows from Eq. (\ref{Gibbs_Fundamental_Relation_Isolated}).
\end{proof}

The statistical demonstration of the identity in Eq.
(\ref{Irreversible_Heat_Work_equality}) is given in Theorem
\ref{Theorem_Wi_Qi_Equivalence}. It is now easy to show that heat and work can
also be used in the first law for the system ($dE(t)=dQ(t)-dW(t)$), since
$dQ(t)-dW(t)\equiv d_{\text{e}}Q(t)-d_{\text{e}}W(t)$. For the medium, it also
holds ($d\widetilde{E}(t)=d\widetilde{Q}(t)-d\widetilde{W}(t)$) as
$d_{\text{i}}\widetilde{Q}(t)=d_{\text{i}}\widetilde{W}(t)=0$. Thus, we can
express the first law for any body by also using generalized heat and work;
see Eq. (\ref{First_Law}).

It was established in Ref. \cite{Gujrati-I} that%
\begin{equation}
dQ(t)=T(t)dS(t). \label{Def_dQ}%
\end{equation}
It is a generalization of $d_{\text{e}}Q(t)\equiv T_{0}d_{\text{e}}S(t)$ to
$dQ(t)$; the latter denotes the \emph{heat added to the system} either through
exchange with its exterior ($d_{\text{e}}Q(t)$) or by dissipative internal
forces within ($d_{\text{i}}Q(t)$). Similarly,
\begin{equation}
dW(t)=d_{\text{e}}W(t)+d_{\text{i}}W(t)=P(t)dV(t)+A(t)d\xi(t) \label{Def_dW}%
\end{equation}
is the generalization of work done by the system: it includes the work done on
its exterior ($d_{\text{e}}W(t)$) and the (internal) work done by dissipative
internal forces ($d_{\text{i}}W(t)$).

To appreciate the importance of the new definition of heat and work, let us
for the moment assume that there is no internal variable and that $P(t)>P_{0}%
$. The irreversible work $d_{\text{i}}W_{0}(t)$ done in $\Sigma_{0}$ by the
pressure difference $\Delta P(t)=P(t)-P_{0}>0$ is
\[
dW_{0}(t)\equiv d_{\text{i}}W_{0}(t)=\Delta P(t)dV(t)>0,
\]
since $dV(t)>0$, and appears as the irreversible work within the system and
results in raising the kinetic energy $dK_{\text{S}}$ of the center-of-mass of
the surface separating $\Sigma$ and $\widetilde{\Sigma}$ and overcoming work
$dW_{\text{fr}}(t)$ done by all sorts of viscous or frictional drag. Thus,
\begin{equation}
d_{\text{i}}W_{0}(t)\equiv dK_{\text{S}}+dW_{\text{fr}}(t).
\label{Dissipation_work}%
\end{equation}
Because of the stochasticity associated with any statistical system, both
energies on the right side dissipate among the particles so as to increase the
entropy and appear in the form of heat ($d_{\text{i}}Q_{0}(t)=d_{\text{i}%
}Q(t)>0$) within the isolated system \cite{Gujrati-Symmetry}. Thus, when there
are irreversible processes going on, it is natural to generalize heat from
$d_{\text{e}}Q(t)\ $in Eq. (\ref{Standard_Heat_Sum}) to include $d_{\text{i}%
}Q(t)=d_{\text{i}}Q_{0}(t)$ and identify $dQ(t)$ as the heat \emph{added to
the system}. Similarly, we need to generalize work from $d_{\text{e}%
}W(t)=P_{0}dV(t)$ to $dW(t)=P(t)dV(t)$ and identify it as work \emph{done by
the system}. In the presence of the internal variable, there is an additional
contribution $(A(t)-A_{0})d\xi(t)=A(t)d\xi(t)>0$ to $dW(t)$. This does not
change the conclusions above.

We finally conclude that
\begin{equation}
dE(t)=d_{\text{e}}Q(t)-d_{\text{e}}W(t)\equiv dQ(t)-dW(t)
\label{General_First_Law}%
\end{equation}
which demonstrates that both formulations are \emph{valid }for the first law.
However, the most important result is given by Eq. (\ref{Def_dQ}). We also see
that $d_{\text{i}}Q(t)\neq T_{0}d_{\text{i}}S(t),\ d_{\text{i}}Q(t)\neq
T(t)d_{\text{i}}S(t),$ even though $d_{\text{e}}Q(t)=T_{0}d_{\text{e}}S(t)$;
see Eq. (\ref{Irreversible_Heat_0}).

The equality in Eq. (\ref{Irreversible_Heat_Work_equality}) can also be
obtained by the use of the Gibbs fundamental relation. We follow the approach
initiated in Ref. \cite{Gujrati-II}, and rewrite $\allowbreak dE(t)$ by
explicitly exhibiting the thermodynamic forces as follows:%
\[
dE(t)=T_{0}d_{\text{e}}S(t)-P_{0}dV(t)+T_{0}d_{\text{i}}S(t)+[T(t)-T_{0}%
]dS(t)+[P_{0}-P(t)]dV(t)-A(t)d\xi(t)
\]
to conclude that
\begin{equation}
T_{0}d_{\text{i}}S(t)+[T(t)-T_{0}]dS(t)+[P_{0}-P(t)]dV(t)-A(t)d\xi(t)=0.
\label{Irreversible-contributions}%
\end{equation}
For this to be valid, each of the last three terms must be non-positive:%
\begin{equation}
\lbrack T_{0}-T(t)]dS(t)\geq0,[P(t)-P_{0}]dV(t)\geq0,A(t)d\xi(t)\geq0,
\label{Irreversible-entropy-contributions}%
\end{equation}
\ to ensure that $d_{\text{i}}S(t)\geq0$. The factors $T_{0}-T(t),$
$P(t)-P_{0}$ and $A(t)$\ in front of the extensive state variables are the
corresponding thermodynamic forces that act to bring the system to
equilibrium. In the process, each force generates its own irreversible entropy
generation \cite{Gujrati-II}. The equalities occur when thermodynamic forces vanish.

It is useful to acknowledge at this point that there are no thermodynamic
forces in the medium. To see this, we consider Eq.
(\ref{Gibbs_Fundamental_Relations})\ for $d\widetilde{E}(t)$,\ in which both
terms contain the constant fields $T_{0}$ and $P_{0}$ of the medium, clearly
showing that the thermodynamic forces are zero. This means that there cannot
be any irreversible entropy generation within the medium; they only appear
within the system, as we have said earlier.

We now recognize that
\begin{subequations}
\label{Irreversible_Heat_0}%
\begin{align}
d_{\text{i}}Q(t)  &  =T_{0}d_{\text{i}}S(t)+[T(t)-T_{0}%
]dS(t),\label{Irreversible_Heat_01}\\
&  =T(t)d_{\text{i}}S(t)+[T(t)-T_{0}]d_{\text{e}}S(t)
\label{Irreversible_Heat_02}%
\end{align}
and
\end{subequations}
\begin{equation}
d_{\text{i}}W(t)=[P(t)-P_{0}]dV(t)+A(t)d\xi(t). \label{Irreversible-Work_0}%
\end{equation}
Their equality merely reflects the fact that in the partition
$dE(t)=d_{\text{e}}E(t)+d_{\text{i}}E(t)$, $d_{\text{i}}E(t)\equiv0$. We also
note that while each term in $d_{\text{i}}W(t)$ is non-negative, this is not
so for $d_{\text{i}}Q(t)$ in which the first term is non-negative, but the
second term in Eq. (\ref{Irreversible_Heat_0}) is non-positive. This not only
means that \emph{the physics of }$d_{\text{i}}Q(t)$\emph{ and }$d_{\text{i}%
}S(t)$\emph{ is very different} but also that
\begin{equation}
d_{\text{i}}Q(t)\leq T_{0}d_{\text{i}}S(t)\,,dQ(t)\leq T_{0}dS(t);
\label{Heat-inequalities}%
\end{equation}
the equalities occur only for isothermal ($T=T_{0}$) or adiabatic ($dS=0$)
processes. Even though $dQ(t)=T(t)dS(t)$ for a system \emph{not} in
equilibrium with the medium, we have $d_{\text{e}}Q(t)\neq T(t)d_{\text{e}%
}S(t),\ d_{\text{i}}Q(t)\neq T(t)d_{\text{i}}S(t).$ \emph{It should become
evident by now that it would be incorrect to use }$dS(t)=d_{\text{e}%
}Q(t)/T_{0}+d_{\text{i}}S(t)$\emph{ to conclude }$d_{\text{i}}Q(t)=T_{0}%
d_{\text{i}}S(t)$\emph{.}

In the general case, the first law can be written as
\[
dE(t)=\frac{\partial E(t)}{\partial\overline{\mathbf{Z}}(t)}\cdot
d\overline{\mathbf{Z}}(t),
\]
where $\overline{\mathbf{Z}}(t)$ contains $S(t)$ and the set $\mathbf{Z}%
_{E}(t)$ consisting of all state variables except $E(t).$ We clearly see that
each term in the scalar product has the same mathematical form, ensuring that
\emph{all terms are on an equal footing}. It is this symmetry that was absent
in the traditional formulation, but is present in the new formulation of the
first law. As this is also the general form of the Gibbs fundamental relation,
\emph{the two laws have reduced to a single law}, as we have claimed.
Therefore, out formulation of the first law will always remain consistent with
the second law.

\section{Generalized Work and the Second Law\label{Sec_Work_Second_Law}}

Let us follow the consequences of this particular generalization a bit further
by again restricting to no internal variable for simplicity, and prove that
only $dW(t)=P(t)dV(t),d\widetilde{W}(t)=P_{0}d\widetilde{V}(t)$ using the
internal pressures of the bodies is consistent with the second law, and not
$dW(t)=P_{0}dV(t),d\widetilde{W}(t)=P(t)d\widetilde{V}(t)$ which use the
pressures external to the bodies. These choices for work are \emph{symmetric}
as opposed to the traditional formulation in Sec.
\ref{sect_Traditional Formulation} in which there is no symmetry between
$d_{\text{e}}W(t)$ and$\ d_{\text{e}}\widetilde{W}(t)$.

We take $P(t)>P_{0}$ and consider the choice $dW(t)=P(t)dV(t)$ etc. We obtain
$dW(t)+d\widetilde{W}(t)=P(t)dV(t)+P_{0}d\widetilde{V}(t)$ valid for any
arbitrary $dV(t)=-d\widetilde{V}(t)$ so that
\begin{equation}
d_{\text{i}}W(t)=\left[  P(t)-P_{0}\right]  dV(t)>0, \label{Work_Forms}%
\end{equation}
which is consistent with the second law and proves the above assertion, once
we recognize that $dV(t)>0$.\ The second choice will result in the violation
the second law, since $dW_{\text{i}}(t)=(P_{0}-P(t))dV(t)<0,$ a physical
impossibility. Thus, we must write the first law for the system and the
medium, respectively, as $dE(t)=dQ(t)-P(t)dV(t),\ d\widetilde{E}%
(t)=d\widetilde{Q}(t)-P_{0}d\widetilde{V}(t).$ The above discussion is easily
extended to include internal variables without affecting the above conclusion.

The above demonstration establishes that the work done by a body is given by
Eq. (\ref{Def_dW}) in \emph{all} cases contrary to the traditional
formulation, see Eq, (\ref{Standard_Heat_Sum}), in which it is given by
$P_{0}dV(t)$. The generalized formulation brings out the another symmetry:
under the interchange system$\Longleftrightarrow$medium, work and heat for any
body always uses its own internal fields. This symmetry is absent in the
traditional formulation. The new symmetry will prove very useful when the
medium is not extensively large compared to the system or when we need to
consider mixing of gases, free expansion, etc. where there is no clear
separation between different parts of an isolated system into a system and a medium.

\section{Statistical Definition of Work and Heat\label{Sec_Stat_Concepts}}

\subsection{System}

\subsubsection{System not in Internal Equilibrium}

Before proceeding further, let us see how the generalized heat and work could
be understood from a statistical point of view. We consider two possible
neighboring nonequilibrium states A and A$^{\prime}$ at different times $t$
and $t^{\prime}<t$, respectively, so that the differences in their state
variables $dE(t)\equiv E(t)-E^{\prime}(t^{\prime}),$ $dV(t)\equiv
V(t)-V^{\prime}(t^{\prime}),d\xi(t)\equiv\xi(t)-\xi^{\prime}(t^{\prime})$ and
the difference $dS(t)\equiv S(t)-S^{\prime}(t^{\prime})$ in their entropies
are infinitesimal. We use the index $i$ to label the microstates of the system
and let $p_{i}(t),p_{i}^{\prime}(t^{\prime})$ their their probabilities in A
and A$^{\prime}$, respectively$.$ These probabilities are functions of the
state variables (including the number of particles, but that remains constant)
and may also have an explicit time dependence. Thus, the discussion here does
not require the system to be in internal equilibrium. Obviously \cite{note-4}
\begin{equation}%
%TCIMACRO{\tsum \nolimits_{i}}%
%BeginExpansion
{\textstyle\sum\nolimits_{i}}
%EndExpansion
p_{i}(t)=%
%TCIMACRO{\tsum \nolimits_{i}}%
%BeginExpansion
{\textstyle\sum\nolimits_{i}}
%EndExpansion
p_{i}^{\prime}(t)\equiv1. \label{Entropy_Sum}%
\end{equation}
We will see below, see Theorem \ref{Theorem_Wi_Qi_Equivalence}, that the
probability conservation is behind the statistical demonstration of the
identity in Eq. (\ref{Irreversible_Heat_Work_equality}). The energy $E_{i}$ of
the $i$th microstate, on the other hand, depends on $V(t)$ and $\xi(t)$ [in
general, $E_{i}$ will depend on the set $\mathbf{Z}_{E}(t)$], but will have no
explicit $t$-dependence. The entropy $S$ and the energy $E$ are given by the
following averages%
\begin{equation}
S(t)\equiv%
%TCIMACRO{\tsum \nolimits_{i}}%
%BeginExpansion
{\textstyle\sum\nolimits_{i}}
%EndExpansion
p_{i}(t)\eta_{i}(t),\ E(t)\equiv%
%TCIMACRO{\tsum \nolimits_{i}}%
%BeginExpansion
{\textstyle\sum\nolimits_{i}}
%EndExpansion
p_{i}(t)E_{i}(t), \label{Entropy_Energy}%
\end{equation}
where
\begin{equation}
\eta_{i}\equiv-\ln p_{i}(t) \label{Uncertainty}%
\end{equation}
is the uncertainty of Shanon or the negative of the index of probability of
Gibbs \cite{Gujrati-Symmetry,Gujrati-III} and $E_{i}(t)$ is the energy of the
$i$th microstate. (We will avoid the use of microstate "entropy," to refer to
$\eta_{i}$, which has become common in the literature these days.) The entropy
expression is due to Gibbs \cite{Gibbs}. We have exhibited a time-dependence
in $E_{i}(t)$\ to reflect the fact that this energy can change as $V(t)$ and
$\xi(t)$ change during the transition A$^{\prime}\rightarrow$A. The microstate
probability will also change in time. In particular,\emph{ a microstate may
disappear or a new microstate may emerge in time. }This is most easily seen by
recognizing, see Eq. (\ref{Internal_Eq_S_S_tilde}), that the entropy is
determined by $W(t).$ As entropy changes, $W(t)$ must change so that either
some previous microstates disappear or some new microstate emerge. All this
will become clear below in Sec. \ref{Sect_Ideal_Gas} where we discuss a simple
example of an ideal gas. We now prove

\begin{theorem}
\label{Theorem_1}$E(t)$ is a function of $V(t),\xi(t)$ and $S(t)$, even though
$E_{i}[V(t),\xi(t)]$ are functions of $V(t)$ and $\xi(t)$\ only.
\end{theorem}

\begin{proof}
We consider the differential
\[
dE(t)\equiv%
%TCIMACRO{\tsum \nolimits_{i}}%
%BeginExpansion
{\textstyle\sum\nolimits_{i}}
%EndExpansion
E_{i}(t)dp_{i}(t)+%
%TCIMACRO{\tsum \nolimits_{i}}%
%BeginExpansion
{\textstyle\sum\nolimits_{i}}
%EndExpansion
p_{i}(t)dE_{i}(t).
\]
As $p_{i}(t)$ are unchanged in the first sum, this sum is evaluated at
\emph{constant entropy}. Thus, this contribution is isentropic which we denote
by $\left.  dE\right\vert _{S}$. The microstate energies $E_{i}$ are unchanged
in the second sum so this contribution refers to an isometric process at fixed
$V(t)$ and $\xi(t)$ and we denote the contribution by $\left.  dE\right\vert
_{V,\xi}$. Thus,%
\begin{equation}
dE\equiv\left.  dE\right\vert _{V,\xi}+\left.  dE\right\vert _{S}.
\label{Energy_Partition}%
\end{equation}
This proves that $E(t)$ is a function of $S(t)$,$V(t)$ and $\xi(t)$.
\end{proof}

In general, $E(t)$ is a function of $S(t)$\ and the set $\mathbf{Z}_{E}(t)$.
We introduce a special process, to be called an \emph{isometric} process,
which is a process at constant $\mathbf{Z}_{E}(t)$ and is a generalization of
an \emph{isochoric} process. In this process, the work done by each mechanical
variables in $\mathbf{Z}_{E}(t)$ remains zero. We now prove the following
theorem that establishes the physical significance of the two contributions.

\begin{theorem}
\label{Theorem_Heat_Work}The isentropic contribution represents the
generalized work $dW(t)$ and the isometric contribution represents the
generalized heat $dQ(t)$.
\end{theorem}

\begin{proof}
We follow Landau and Lifshitz \cite{Landau} and rewrite the second term in Eq.
(\ref{Energy_Partition}) as%
\[
\left.  dE\right\vert _{S}\equiv%
%TCIMACRO{\tsum \nolimits_{i}}%
%BeginExpansion
{\textstyle\sum\nolimits_{i}}
%EndExpansion
p_{i}(t)\frac{\partial E_{i}}{\partial V}dV(t)+%
%TCIMACRO{\tsum \nolimits_{i}}%
%BeginExpansion
{\textstyle\sum\nolimits_{i}}
%EndExpansion
p_{i}(t)\frac{\partial E_{i}}{\partial\xi}d\xi(t)=-%
%TCIMACRO{\tsum \nolimits_{i}}%
%BeginExpansion
{\textstyle\sum\nolimits_{i}}
%EndExpansion
p_{i}(t)P_{i}(t)dV(t)-%
%TCIMACRO{\tsum \nolimits_{i}}%
%BeginExpansion
{\textstyle\sum\nolimits_{i}}
%EndExpansion
p_{i}(t)A_{i}(t)d\xi(t)
\]
where we have introduced $P_{i}(t)\equiv-\partial E_{i}(t)/\partial V(t)$\ as
the pressure produced by the $i$th microstate of the system \cite[p.
67]{Landau-QM} on its boundary. This pressure corresponds to a force pointing
\emph{away} from the system. Similarly, $A_{i}(t)\equiv-\partial
E_{i}(t)/\partial\xi(t)$\ as the affinity of the $i$th microstate. We assume
that the changes $dV(t)$ and $d\xi(t)$\ are the same for all microstates so
that they can be taken out of the summations. In terms of
\begin{equation}
P(t)=%
%TCIMACRO{\tsum \nolimits_{i}}%
%BeginExpansion
{\textstyle\sum\nolimits_{i}}
%EndExpansion
p_{i}(t)P_{i}(t),A(t)=%
%TCIMACRO{\tsum \nolimits_{i}}%
%BeginExpansion
{\textstyle\sum\nolimits_{i}}
%EndExpansion
p_{i}(t)A_{i}(t), \label{Statistical Fields}%
\end{equation}
respectively, which define the instantaneous \emph{average pressure}
$P(t)$\emph{ and affinity} $A(t)$\emph{ }of the system, respectively, we can
relate $\left.  dE\right\vert _{S}$ with $dW(t)$ and nit with $d_{\text{e}%
}W(t)$:
\begin{equation}
dW(t)\equiv-\left.  dE\right\vert _{S}=P(t)dV(t)+A(t)d\xi(t). \label{dE_dW}%
\end{equation}
This identification then also proves that the heat in the first law must be
properly identified with $dQ$ and not with $d_{\text{e}}Q$. Accordingly,
\begin{equation}
dQ\equiv\left.  dE\right\vert _{V,\xi}\equiv%
%TCIMACRO{\tsum \nolimits_{i}}%
%BeginExpansion
{\textstyle\sum\nolimits_{i}}
%EndExpansion
E_{i}dp_{i}, \label{dE_dQ}%
\end{equation}
i.e., $dQ$ for irreversible processes is \emph{the isometric change in the
energy}.
\end{proof}

We should point out that by assuming $dV(t)$ and $d\xi(t)$ above to be the
same for all microstates, the statistical nature of $\left.  dE\right\vert
_{S}$ is reflected in the statistical nature of $P(t)$ and $A(t)$, the
internal fields of the system. Thus, the internal fields are fluctuating
quantities from microstate to microstate when $dV(t)$ and $d\xi(t)$ are not
treated statistically.

In general, $\left.  dE\right\vert _{S}$ will be a sum of various works
$dW_{Z}(t)=-(\partial E/\partial Z)_{S,\mathbf{Z}_{E}^{\prime}}dZ$, with
$\mathbf{Z}_{E}^{\prime}(t)$ consisting of all state variables in
$\mathbf{Z}_{E}(t)$ except $Z(t)$ used in the derivative:%
\begin{equation}
\left.  dE\right\vert _{S}\equiv-dW(t)=-\sum_{Z\in\mathbf{Z}_{E}(t)}dW_{Z}(t).
\label{General-Work}%
\end{equation}
This is again consistent with the previously mentioned symmetry in the new formulation.

The above discussion proves that the definition of heat and work does not
require the establishment of the internal equilibrium within the system.\ It
is useful to compare the above approach with the traditional formulation of
the first law in terms of $d_{\text{e}}Q(t)$ and $d_{\text{e}}W(t)$:
\emph{both formulations are valid in all cases}. It should be mentioned that
the above identification is well known in equilibrium statistical mechanics,
but its extension to irreversible processes and our interpreation is, to the
best of our knowledge, is novel.

We now prove a trivial but conceptually an important theorem.

\begin{theorem}
\label{Theorem_Traditional_Heat_Work}Heat and work in the traditional
formulation of the first law do not have a clear division as in Eq.
(\ref{Energy_Partition}).
\end{theorem}

\begin{proof}
To prove the theorem, we focus on $d_{\text{e}}W(t)=P_{0}dV(t)$ and recognize
that $P_{0}$ is a constant. Hence, the statistical nature of $d_{\text{e}%
}W(t)$ must be contained in $dV(t)$. As $V(t)=\sum_{i}p_{i}V_{i}$, we have
\[
d_{\text{e}}W(t)=P_{0}(%
%TCIMACRO{\tsum \nolimits_{i}}%
%BeginExpansion
{\textstyle\sum\nolimits_{i}}
%EndExpansion
p_{i}dV_{i}+%
%TCIMACRO{\tsum \nolimits_{i}}%
%BeginExpansion
{\textstyle\sum\nolimits_{i}}
%EndExpansion
V_{i}dp_{i}).
\]
We observe that $d_{\text{e}}W(t)$ contains not only an isentropic
contribution, the first sum on the right, but also contains a contribution
containing $dp_{i}$. Thus, the clear separation between heat and work, as
appears in Theorem \ref{Theorem_1}, is not present in the traditional
formulation of the first law.
\end{proof}

\subsubsection{System in Internal Equilibrium}

It should be clear from the existence of non-zero thermodynamic forces for
irreversibility, that $P(t)\neq$ $P_{0}=\left(  \partial\widetilde{E}/\partial
V\right)  _{\widetilde{S},\widetilde{\xi}}$ except when a mechanical
equilibrium exists. While the instantaneous average pressure is defined under
all circumstances, it can only be identified with the thermodynamic definition
of the instantaneous pressure%
\begin{equation}
P(t)=-\left(  \partial E/\partial V\right)  _{S,\xi} \label{Pressure}%
\end{equation}
when the system is in internal equilibrium. Similarly, the instantaneous
average affinity $A(t)$ of the state has nothing to do with the affinity of
the medium ($A_{0}=0$), and can only be identified with its thermodynamic
definition $A(t)=-(\partial E/d\xi)_{S,V}$\ when the system is in internal
equilibrium. To proceed further, we need to impose the condition of internal
equilibrium, so that $p_{i}$\ has no explicit time-dependence. In this case,
not only the instantaneous pressure satisfies Eq. (\ref{Pressure}), but we
also have, following Eq. (\ref{dE_dQ}),
\begin{equation}
dQ(t)/dS(t)=\left(  \partial E/\partial S\right)  _{V,\xi}=T(t),
\label{System_dQ_dS}%
\end{equation}
which is a statistical proof of the thermodynamic identity in Eq.
(\ref{Def_dQ}) relating $dQ(t)$ and $dS(t)$. We also note that the ratio
$dQ(t)/dS(t)$ is related to a field variable of a (macro)state, the
instantaneous temperature of the system, while in the conventional approach,
the ratio $d_{\text{e}}Q(t)/d_{\text{e}}S(t)=T_{0}$ does not give a field
variable of the state.

It is clear from the above discussion that it is heat and not work that causes
$p_{i}(t)$, and therefore the entropy, to change without changing $E_{i}$.
This is the essence of the common wisdom that heat is \emph{random motion}.
But we now have a mathematical definition: heat is the isometric part of
$dE(t)$ that is directly related to the change in the entropy through changes
in $p_{i}(t)$. Work is that part of the energy change caused by isentropic
variations in the "mechanical" state variables $\mathbf{Z}_{E}(t)$. Thus, work
causes $E_{i}$ to change without changing $p_{i}(t)$. This is true no matter
how far the system is from equilibrium. Thus, our formulation of the first law
and the identification of the two terms is the most general one, and
applicable in all cases. \emph{The relationship between heat and entropy
becomes simple only when the system is also in internal equilibrium in which
case }$T(t)$\emph{ has a thermodynamic significance; see Eq.
(\ref{Thermodynamic-Fields-System}) and we have the thermodynamic identity in
Eq. (\ref{Def_dQ}) relating }$dQ(t)$\emph{ and }$dS(t)$\emph{. }

\subsection{Microstate probabilities}

In internal equilibrium, the entropy must be at its maximum at fixed
$E(t)=\sum_{i}E_{i}p_{i},V(t)=\sum_{i}V_{i}p_{i}$ and $\xi(t)=\sum_{i}\xi
_{i}p_{i}$, and is obtained by varying $p_{i}$ without changing the
microstates, i.e. $E_{i},V_{i}$ and $\xi_{i}$.\ This variation has nothing to
do with $dp_{i}$\ in a physical process. Using the Lagrange multiplier
technique, it is easy to show that the condition for this in terms of four
Lagrange multipliers whose definitions are obvious is
\begin{equation}
\eta_{i}=\lambda_{1}+\lambda_{2}E_{i}+\lambda_{3}V_{i}+\lambda_{4}\xi_{i},
\label{index_i}%
\end{equation}
from which follows $S=\lambda_{1}+\lambda_{2}E+\lambda_{3}V+\lambda_{4}\xi$.
It is now easy to identify $\lambda_{2}=\beta,\lambda_{3}=\beta P$ and
$\lambda_{4}=\beta A$ so we finally have
\begin{equation}
p_{i}(t)=\exp[\beta(t)(\widehat{G}(t)-E_{i}-P(t)V_{i}-A(t)\xi_{i})],
\label{microstate probability}%
\end{equation}
where $\lambda_{1}=\beta(t)\widehat{G}(t)$\ with $\widehat{G}(t)$ defined by
\[
\exp(-\beta(t)\widehat{G}(t))\equiv\sum_{i}\exp[-\beta(t)(E_{i}+P(t)V_{i}%
+A(t)\xi_{i})].
\]
The quantity $\widehat{G}(t)$ would represent the free energy of the system,
had it been in a medium $\widetilde{\Sigma}(T,P,A)$. However, as the system is
in a medium $\widetilde{\Sigma}(T_{0},P_{0},A_{0}=0),$ $\widehat{G}(t)$ does
not represent the free energy in this case; the correct free energy of the
system is the Gibbs free energy $G(T_{0},P_{0})=E(t)-T_{0}S(t)+P_{0}V(t)$; see
Ref. \cite{Gujrati-II} for more details. The microstate probability $p_{i}(t)$
in Eq. (\ref{microstate probability}) clearly shows the effect of
irreversibility and is very different from its equilibrium analog
$p_{i\text{,eq}}$
\[
p_{i\text{,eq}}=\exp[\beta_{0}(G(T_{0},P_{0})-E_{i}-P_{0}V_{i})].
\]

We now provide another demonstration that the two terms in Eq.
(\ref{Energy_Partition}) are identical to the generalized heat and work in the
following theorem.

\begin{theorem}
\label{Theorem_Work_Identification}The heat and work in Eqs. (\ref{Def_dQ})
and (\ref{Def_dW}) , and (\ref{Energy_Partition}) are the same.
\end{theorem}

\begin{proof}
We first note the identity%
\[%
%TCIMACRO{\tsum \nolimits_{i}}%
%BeginExpansion
{\textstyle\sum\nolimits_{i}}
%EndExpansion
(\eta_{i}-\beta(t)E_{i})dp_{i}=0,
\]
which is nothing but the identity $dQ(t)=T(t)dS(t).$ Using this identity, we
obtain another identity that follows from Eq. (\ref{index_i})
\begin{equation}
P(t)%
%TCIMACRO{\tsum \nolimits_{i}}%
%BeginExpansion
{\textstyle\sum\nolimits_{i}}
%EndExpansion
V_{i}dp_{i}(t)+A(t)%
%TCIMACRO{\tsum \nolimits_{i}}%
%BeginExpansion
{\textstyle\sum\nolimits_{i}}
%EndExpansion
\xi_{i}dp_{i}(t)=0, \label{Work_Identity}%
\end{equation}
where being related to the Lagrange multipliers, $P(t),A(t)$, etc. must be
treated as parameters. We\ now rewrite $dV(t)=%
%TCIMACRO{\tsum \nolimits_{i}}%
%BeginExpansion
{\textstyle\sum\nolimits_{i}}
%EndExpansion
V_{i}(t)dp_{i}(t)+%
%TCIMACRO{\tsum \nolimits_{i}}%
%BeginExpansion
{\textstyle\sum\nolimits_{i}}
%EndExpansion
p_{i}(t)dV_{i}(t)$ and $d\xi(t)=%
%TCIMACRO{\tsum \nolimits_{i}}%
%BeginExpansion
{\textstyle\sum\nolimits_{i}}
%EndExpansion
\xi_{i}(t)dp_{i}(t)+%
%TCIMACRO{\tsum \nolimits_{i}}%
%BeginExpansion
{\textstyle\sum\nolimits_{i}}
%EndExpansion
p_{i}(t)d\xi_{i}(t)$ in $dW(t)=P(t)dV(t)+A(t)d\xi(t)$ and use the above
identity in Eq. (\ref{Work_Identity}) to establish that the work in Eqs.
(\ref{Def_dW}) reduces to the isentropic form $dW(t)=P(t)%
%TCIMACRO{\tsum \nolimits_{i}}%
%BeginExpansion
{\textstyle\sum\nolimits_{i}}
%EndExpansion
p_{i}(t)dV_{i}+A(t)%
%TCIMACRO{\tsum \nolimits_{i}}%
%BeginExpansion
{\textstyle\sum\nolimits_{i}}
%EndExpansion
p_{i}(t)d\xi_{i}(t)$. As $P(t)$ and $A(t)$ are parameters, the statistical
nature of $dW(t)$ appears in $dV(t)$ and $d\xi(t)$. As the work expressions
are identical, the demonstration also proves that the thermodynamic parameters
in Eq. (\ref{microstate probability}) and the statistical fields in Eq.
(\ref{Statistical Fields}) are the same.
\end{proof}

This proves the two formulations of work to be identical. Apart from providing
a check for the internal consistency in our approach, this also proves that
the first term in Eq. (\ref{Energy_Partition}) is identical to the heat in Eq.
(\ref{Def_dQ}). At equilibrium, $A(t)\rightarrow A_{0}=0$ and the second term
in Eq. (\ref{Work_Identity}) vanishes. As $P(t)\rightarrow P_{0}$ at
equilibrium, we conclude that $d_{\text{e}}W(t)=P_{0}%
%TCIMACRO{\tsum \nolimits_{i}}%
%BeginExpansion
{\textstyle\sum\nolimits_{i}}
%EndExpansion
p_{i}dV_{i}$, see Theorem \ref{Theorem_Traditional_Heat_Work}, which makes
$d_{\text{e}}W(t)$ purely isentropic, as expected.

Let us compute the generalized pressure work $W(t)$ done by the system over
some time interval $t$ at \emph{fixed} $p_{i}(t)$\ as $V(t)$ [$\overset{.}%
{V}(t)\equiv dV(t)/dt$] changes from its initial value. It is given by the
integral%
\begin{equation}
W_{\text{V}}(t)=%
%TCIMACRO{\tsum \nolimits_{i}}%
%BeginExpansion
{\textstyle\sum\nolimits_{i}}
%EndExpansion
p_{i}(t)\int_{0}^{t}p_{i}(t)P_{i}(t)\overset{.}{V}(t)dt\equiv\int_{0}%
^{t}P(t)\overset{.}{V}(t)dt \label{Work_system}%
\end{equation}
and differs from the negative of the work $\widetilde{W}(t)=P_{0}%
\Delta\widetilde{V}(t)$ (obtained by replacing all the quantities in Eq.
(\ref{Work_system}) by their medium analogs) done by the medium on the system
by the dissipative contribution. This result should contrasted with the result
obtained by others \cite{note-2}. They become identical if and only if the
average pressure $P(t)$ is \emph{identical} with $P_{0}$.

\subsection{Medium}

The above discussion can be easily extended to the medium (the suffix
$\widetilde{\alpha}$\ denotes its microstates) with the following results%
\begin{align*}
d\widetilde{W}(t)  &  =-\left.  d\widetilde{E}\right\vert _{\widetilde{S}%
}\equiv-%
%TCIMACRO{\tsum \nolimits_{\widetilde{\alpha}}}%
%BeginExpansion
{\textstyle\sum\nolimits_{\widetilde{\alpha}}}
%EndExpansion
\widetilde{p}_{\widetilde{\alpha}}\frac{\partial\widetilde{E}_{\widetilde
{\alpha}}}{\partial\widetilde{V}}d\widetilde{V}=P_{0}d\widetilde{V},\\
d\widetilde{Q}(t)  &  =\left.  d\widetilde{E}\right\vert _{\widetilde{V}%
}\equiv%
%TCIMACRO{\tsum \nolimits_{\widetilde{\alpha}}}%
%BeginExpansion
{\textstyle\sum\nolimits_{\widetilde{\alpha}}}
%EndExpansion
\widetilde{E}_{\widetilde{\alpha}}d\widetilde{p}_{\widetilde{\alpha}}\equiv
d\widetilde{Q}\equiv-d_{\text{e}}Q,
\end{align*}
where all the quantities refer to the medium, except $d_{\text{e}}Q,$ and have
their standard meaning. The analog of Eq. (\ref{System_dQ_dS}) is%
\begin{equation}
d\widetilde{Q/}d\widetilde{S}=\left(  \partial\widetilde{E}/\partial
\widetilde{S}\right)  _{\widetilde{V},\widetilde{\xi}}=T_{0},
\label{Medium_dQ_dS}%
\end{equation}
as expected. We clearly see that
\[
dW+d\widetilde{W}\neq0
\]
such as when mechanical equilibrium is not present. In this case, we also have%
\[
dQ+d\widetilde{Q}\neq0.
\]

\subsection{Irreversible Work and Heat \ \ }

We can now identify $d_{\text{i}}W(t)$ and $d_{\text{i}}Q(t):$%
\begin{align*}
d_{\text{i}}W(t)  &  \equiv-(\left.  dE\right\vert _{S}+\left.  d\widetilde
{E}\right\vert _{\widetilde{S}})\\
d_{\text{i}}Q(t)  &  \equiv(\left.  dE\right\vert _{V,\xi}+\left.
d\widetilde{E}\right\vert _{\widetilde{V},\widetilde{\xi}}),
\end{align*}
with $d_{\text{i}}Q(t)\equiv d_{\text{i}}W(t)$ from Theorem
\ref{Theorem_Irreversible-Heat-Work}. Because of the equality, we only needs
to compute one of them, which we take to be $d_{\text{i}}W(t)$, merely because
it only involves adiabatic quantities. However, the computation of
irreversible work requires considering both the system and the (working part
of the) medium, which makes their computation quite unfeasible in many
situations because of the very large size of the medium, unless the equations
of state of the medium are known. On the other hand, the determination of
$dW(t)$ and $dQ(t)$ is a much easier task computationally as we only deal with
the system, a point made several times in the literature; see, for example,
Jarzynski \cite{Jarzynski}.

\subsection{External and Internal Variations of $dp_{i}(t)$%
\label{Sec_External_Internal_Variations}}

Let us introduce the natural partition as in Eq. (\ref{Generalized-Heat-Work})
for $dp_{i}$:%
\[
dp_{i}(t)=d_{\text{e}}p_{i}(t)+d_{\text{i}}p_{i}(t);
\]
$d_{\text{e}}p_{i}(t)$ is the change due to exchanges with the medium and
$d_{\text{i}}p_{i}(t)$ the change due to internal dissipation $d_{\text{i}}Q$.
Then we have%
\begin{equation}
d_{\text{e}}Q(t)\equiv%
%TCIMACRO{\tsum \nolimits_{i}}%
%BeginExpansion
{\textstyle\sum\nolimits_{i}}
%EndExpansion
E_{i}d_{\text{e}}p_{i}(t),\ d_{\text{i}}Q(t)\equiv%
%TCIMACRO{\tsum \nolimits_{i}}%
%BeginExpansion
{\textstyle\sum\nolimits_{i}}
%EndExpansion
E_{i}d_{\text{i}}p_{i}(t).\ \label{Statistical_Heat_Components}%
\end{equation}
We similarly have%
\begin{equation}
d_{\text{e}}S(t)\equiv%
%TCIMACRO{\tsum \nolimits_{i}}%
%BeginExpansion
{\textstyle\sum\nolimits_{i}}
%EndExpansion
\eta_{i}d_{\text{e}}p_{i}(t),\ d_{\text{i}}S(t)\equiv%
%TCIMACRO{\tsum \nolimits_{i}}%
%BeginExpansion
{\textstyle\sum\nolimits_{i}}
%EndExpansion
\eta_{i}d_{\text{i}}p_{i}(t).\ \label{Statistical_Entropy_Components}%
\end{equation}
We can recast Eq. (\ref{Irreversible_Heat_0})
\[%
%TCIMACRO{\tsum \nolimits_{i}}%
%BeginExpansion
{\textstyle\sum\nolimits_{i}}
%EndExpansion
(E_{i}-T_{0}\eta_{i})d_{\text{i}}p_{i}=(T(t)-T_{0})%
%TCIMACRO{\tsum \nolimits_{i}}%
%BeginExpansion
{\textstyle\sum\nolimits_{i}}
%EndExpansion
\eta_{i}dp_{i}%
\]
that acts as a constraint on possible variations $d_{\text{i}}p_{i}$.\ The
relation $d_{\text{e}}Q(t)=T_{0}d_{\text{e}}S(t)$ can be expressed in terms of
$d_{\text{e}}p_{i}(t)$%
\[%
%TCIMACRO{\tsum \nolimits_{i}}%
%BeginExpansion
{\textstyle\sum\nolimits_{i}}
%EndExpansion
(\eta_{i}-\beta_{0}E_{i})d_{\text{e}}p_{i}=0.
\]

We now prove the following theorem.

\begin{theorem}
\label{Theorem_Wi_Qi_Equivalence}The identity in Eq.
(\ref{Irreversible_Heat_Work_equality}) is a consequence of the vanishing of
average change in microstate uncertainty
\[%
%TCIMACRO{\tsum \nolimits_{i}}%
%BeginExpansion
{\textstyle\sum\nolimits_{i}}
%EndExpansion
p_{i}d\eta_{i}=%
%TCIMACRO{\tsum \nolimits_{i}}%
%BeginExpansion
{\textstyle\sum\nolimits_{i}}
%EndExpansion
dp_{i}=0,
\]
which is nothing but the conservation of probability.
\end{theorem}

\begin{proof}
We express Eq. (\ref{Irreversible_Heat_Work_equality}) in terms of microstate
probabilities and use Eq. (\ref{Work_Identity}) to obtain%
\[%
%TCIMACRO{\tsum \nolimits_{i}}%
%BeginExpansion
{\textstyle\sum\nolimits_{i}}
%EndExpansion
E_{i}d_{\text{i}}p_{i}(t)=P(t)%
%TCIMACRO{\tsum \nolimits_{i}}%
%BeginExpansion
{\textstyle\sum\nolimits_{i}}
%EndExpansion
p_{i}(t)dV_{i}+A(t)%
%TCIMACRO{\tsum \nolimits_{i}}%
%BeginExpansion
{\textstyle\sum\nolimits_{i}}
%EndExpansion
p_{i}(t)d\xi_{i}(t)-P_{0}dV.
\]
We eliminate the last term using the traditional formulation of the first law
to obtain%
\[%
%TCIMACRO{\tsum \nolimits_{i}}%
%BeginExpansion
{\textstyle\sum\nolimits_{i}}
%EndExpansion
p_{i}(t)\left[  dE_{i}(t)+P(t)dV_{i}+A(t)d\xi_{i}(t)\right]  =0.
\]
In view of Eq. (\ref{index_i}), this is nothing but $%
%TCIMACRO{\tsum \nolimits_{i}}%
%BeginExpansion
{\textstyle\sum\nolimits_{i}}
%EndExpansion
p_{i}d\eta_{i}=0$. This proves the theorem.
\end{proof}

\subsection{The Adiabatic Theorem}

We now have a clear statement of the generalization of the adiabatic theorem
\cite{Landau} for nonequilibrium processes. \emph{An adiabatic nonequilibrium
process is an isentropic process. } Such a process also includes the
stationary limit, i.e. the steady state of a non-equilibrium process. However,
the extension goes beyond the conventional notion of an adiabatic process
commonly dealt with in equilibrium statistical mechanics according to which an
adiabatic process is one for which $d_{\text{i}}S=0$, and represents a
reversible process in a thermally isolated system so that $d_{\text{e}}%
Q(t)=0$. One can also have $dS=0$ in an irreversible process during which%
\begin{equation}
d_{\text{i}}S(t)=-d_{\text{e}}S(t)>0;\label{Stationary_S}%
\end{equation}
as usual, Eq. (\ref{Irreversible_Heat_Work_equality}) always remains
satisfied. If the system remains in internal equilibrium, which may not hold
for a stationary state, then we must also have%
\begin{equation}
d_{\text{i}}Q(t)=-d_{\text{e}}Q(t)=T_{0}d_{\text{i}}%
S(t)>0.\label{Stationary _Q}%
\end{equation}
In the following, we will not consider a stationary nonequilibrium state, a
case that would be taken up elsewhere.

\begin{theorem}
\label{Adiabatic Theorem} In an adiabatic process, the sets of microstates and
of their probabilities $p_{i}$ do not change, but $d_{\text{e}}p_{i}%
=-d_{\text{i}}p_{i}\neq0$ for all $i$.
\end{theorem}

\begin{proof}
In terms $d_{\text{e}}p_{i}$ and $d_{\text{i}}p_{i}$, Eqs. (\ref{Stationary_S}%
) and (\ref{Stationary _Q}) become
\begin{subequations}
\label{Adiabatic_entropy_heat}%
\begin{align}
\sum_{i}\eta_{i}d_{\text{i}}p_{i} &  =-\sum_{i}\eta_{i}d_{\text{e}}%
p_{i},\label{adiabatic_entropy}\\
\sum_{i}E_{i}d_{\text{i}}p_{i} &  =-\sum_{i}E_{i}d_{\text{e}}p_{i}%
.\label{adiabatic_heat}%
\end{align}
Recognizing that there is only work in the defining relation in
Eq.(\ref{Energy_Partition}), which requires $p_{i}$ not to change, we conclude
that
\end{subequations}
\[
dp_{i}=0\text{ \ for }\forall i
\]
in an adiabatic process. As $d_{\text{i}}S(t)$ does not vanish in an
irreversible process, $d_{\text{i}}p_{i}(t)$ cannot vanish. Accordingly,
$d_{\text{e}}p_{i}=-d_{\text{i}}p_{i}\neq0$ for an irreversible adiabatic
process. As $p_{i}$'s do not change, no microstate can appear or disappear.
This proves the theorem.
\end{proof}

\subsection{Ideal Gas:\ An Illustrative Example\label{Sect_Ideal_Gas}}

Consider, as an example, an ideal gas in a cuboid with its axis along the
$x$-axis. Its length $a(t)$ along the $x$-axis is controlled by a frictionless
piston at one end. All other walls are assumed rigid. Let the dimensions $b$
and $c$, respectively, along the $y$- and $z$- axes remain constant.\ As the
particles are ideal, we can focus on single particle energy levels given by
\[
E_{\mathbf{n}}=\frac{\hbar^{2}\pi^{2}}{2m}(\frac{n_{1}^{2}}{a^{2}(t)}%
+\frac{n_{2}^{2}}{b^{2}}+\frac{n_{3}^{2}}{c^{2}}),
\]
with the corresponding wavefunction given by
\[
\phi_{\mathbf{n}}(\mathbf{r}(t))=\sqrt{\frac{8}{a(t)bc}}\sin(\frac{n_{1}}{\pi
a(t)}x)\sin(\frac{n_{2}}{\pi b}y)\sin(\frac{n_{3}}{\pi c}z),
\]
which form the normalized complete basis set. Let $p_{\mathbf{n}}(t)$ denote
the probability of having gas particles in a given eigenstate $\mathbf{n}$.
Then the energy per particle of the gas is given by
\[
E(t)=\sum_{\mathbf{n}}p_{\mathbf{n}}(t)E_{\mathbf{n}}(t),
\]
and the pressure on the moving piston is given by
\[
P(t)=\frac{1}{bc}(-\frac{dE}{da(t)})=\frac{\hbar^{2}\pi^{2}}{mV(t)}%
\sum_{\mathbf{n}}p_{\mathbf{n}}(t)\frac{n_{1}^{2}}{a^{2}(t)}=\frac{2}%
{V(t)}(E(t)-E_{0}),
\]
where
\[
E_{0}=\sum_{\mathbf{n}}p_{\mathbf{n}}(t)\frac{\hbar^{2}\pi^{2}}{2m}%
(\frac{n_{2}^{2}}{b^{2}}+\frac{n_{3}^{2}}{c^{2}}).
\]
The probabilities and the energy eigenvalues change in time as a function of
$a(t)$. The change in the energy\ $E(t)$ comes from changes in $p_{\mathbf{n}%
}(t)$ and from the variations in $a(t)$. The two contributions determine heat
and work, respectively.

If the system of ideal gas is isolated, its energy, volume and the number of
particles remain constant. If the gas is originally not in equilibrium, it
will eventually reach equilibrium in which its entropy must increase. This
requires the introduction of some internal variables even in this system whose
variation will give rise to entropy generation by causing internal variations
$d_{\text{i}}p_{\mathbf{n}}(t)$ in $p_{\mathbf{n}}(t)$. Here, we will assume a
single internal variable $\xi(t)$. Its real nature is not relevant for our
discussion. What is relevant is that the variation in $\xi(t)$ is accompanied
by changes $d_{\text{i}}p_{\mathbf{n}}(t)$ occurring within the isolated
system. According to our identification of heat with changes in $p_{\mathbf{n}%
}(t)$, these variations must be associated with heat, which in this case will
be associated with irreversible heat $d_{\text{i}}Q(t)$.

\subsubsection{Isothermal Expansion}

Let us first consider an isothermal expansion of this gas so that the
temperature of the gas remains constant and equal to that of the medium
$T_{0}$. During expansion, energy is pumped into the gas isothermally from
outside so that not only $E(t)$ remains constant, but also keeps $P(t)V(t)$.
The pumping of energy will result in the change $d_{\text{e}}p_{\mathbf{n}%
}(t)$. This will determine $d_{\text{e}}Q(t)=T_{0}d_{\text{e}}S(t)$. In
addition, particles may undergo transitions among various energy levels, as
discussed above, without any external energy input, which will determine the
change $d_{\text{i}}p_{\mathbf{n}}(t)$. This will determine $d_{\text{i}%
}Q(t)=T_{0}d_{\text{i}}S(t)$, and consequently $d_{\text{i}}W(t)$. Thus,%
\[
\lbrack P(t)-P_{0}]dV(t)+A(t)d\xi(t)=T_{0}dS(t)-d_{\text{e}}Q(t),
\]
which allows us to determine the irreversible work in terms of measurable
quantities. Such a calculation will not be possible in the traditional
formulation of the first law.

\subsubsection{Adiabatic Expansion}

In a nonequilibrium adiabatic process, we have $d_{\text{i}}W(t)=-$
$d_{\text{e}}Q(t)$ so the heat exchange $\left\vert d_{\text{e}}%
Q(t)\right\vert =$ $T_{0}\left\vert d_{\text{e}}S(t)\right\vert $ is converted
into the irreversible work in this process. We can use this to determine the
work $d_{\text{i}}W_{\xi}(t)$ due to the internal variable%
\[
A(t)d\xi(t)=-d_{\text{e}}Q(t)-(P(t)-P_{0})dV>0.
\]
The identification $d_{\text{i}}W(t)=-$ $d_{\text{e}}Q(t)$ and the calculation
of $A(t)d\xi(t)$ cannot be done in the traditional formulation of the first law.

\section{Clausius Equality \label{Sec_Clausius_Equality}}

It follows from Eq. (\ref{Def_dQ}) that $dQ(t)/T(t)$ is nothing but the
\emph{exact differential }$dS(t)$ so that
\begin{equation}
\oint dQ(t)/T(t)\equiv0\ \ \text{for}\ \text{any cyclic process.}
\label{Clausius_Equality_0}%
\end{equation}
We will call it the \emph{Clausius equality}. The equality should not be
interpreted as the absence of irreversibility; see Eq.
(\ref{Clausius_Inequality_00}). It is only because of the use of the
generalized heat $dQ(t)$ in place of $d_{\text{e}}Q(t$\ that the Clausius
inequality has become an equality. Using $d_{\text{i}}S(t)\equiv
dS(t)-d_{\text{e}}S(t)$, we obtain the original Clausius inequality for a
cyclic process taking time $\tau$:
\begin{equation}
N(t,\tau)\equiv%
%TCIMACRO{\toint }%
%BeginExpansion
{\textstyle\oint}
%EndExpansion
d_{\text{i}}S(t)=-%
%TCIMACRO{\toint }%
%BeginExpansion
{\textstyle\oint}
%EndExpansion
d_{\text{e}}Q(t)/T_{0}\geq0, \label{Irreversible_entropy_Cycle}%
\end{equation}
which is the second law for a cyclic process and represents the irreversible
entropy generated in a cycle; see Eq. (\ref{Clausius_Inequality_00}). The
quantity $N(t)$ is the \emph{uncompensated transformation} of Clausius
\cite{Prigogine} that, as we have just discovered, is directly related to
$d_{\text{i}}S(t)$ \cite{Eu}; in contrast, $N_{0}(t,\tau)$
\begin{equation}
N_{0}(t,\tau)\equiv%
%TCIMACRO{\toint }%
%BeginExpansion
{\textstyle\oint}
%EndExpansion
d_{\text{i}}Q(t)/T(t) \label{Irreversible_entropy_Cycle-0}%
\end{equation}
is determined by the \emph{uncompensated heat} $d_{\text{i}}Q(t)$ and
represents a different quantity. In terms of the two heats, we have
\begin{equation}%
%TCIMACRO{\toint }%
%BeginExpansion
{\textstyle\oint}
%EndExpansion
d_{\text{i}}Q(t)/T(t)=-%
%TCIMACRO{\toint }%
%BeginExpansion
{\textstyle\oint}
%EndExpansion
d_{\text{e}}Q(t)/T(t)\geq0, \label{Clausius_Inequality}%
\end{equation}
which results in a new Clausius inequality\emph{ }%
\begin{equation}%
%TCIMACRO{\toint }%
%BeginExpansion
{\textstyle\oint}
%EndExpansion
d_{\text{e}}Q(t)/T(t)\leq0; \label{Clausius_Inequality_0}%
\end{equation}
compare with the original Clausius inequality in Eq.
(\ref{Clausius_Inequality_00}). The contribution $N_{0}(t)$ in Eq.
(\ref{Clausius_Inequality}) can be thought of as a new uncompensated
transformation, different from $N(t)$ \cite{Prigogine}. Our formulation has
allowed us to identify the two transformations%
\[
N(t,\tau)\equiv%
%TCIMACRO{\toint }%
%BeginExpansion
{\textstyle\oint}
%EndExpansion
d_{\text{i}}S(t)\equiv%
%TCIMACRO{\toint }%
%BeginExpansion
{\textstyle\oint}
%EndExpansion
d_{\text{i}}Q(t)/T_{0}-%
%TCIMACRO{\toint }%
%BeginExpansion
{\textstyle\oint}
%EndExpansion
T(t)dS(t)/T_{0},~N_{0}(t,\tau)\equiv%
%TCIMACRO{\toint }%
%BeginExpansion
{\textstyle\oint}
%EndExpansion
d_{\text{i}}Q(t)/T(t)\equiv%
%TCIMACRO{\toint }%
%BeginExpansion
{\textstyle\oint}
%EndExpansion
d_{\text{i}}W(t)/T(t),
\]
which provides a way of computing them using nonequilibrium equations of state
of $\Sigma$ and $\widetilde{\Sigma}$. \ \ 

Other quantities also appear as equalities. Let us consider work, which
appears as $dW(t)=-\left[  dE(t)-T(t)dS(t)\right]  $. For a not too small a
change, we have
\begin{equation}
\Delta W(t)=-\Delta\widehat{F}(t)\equiv\Delta\left[  E(t)-T(t)S(t)\right]  .
\label{Work_FreeEnergy_Equality}%
\end{equation}
The quantity $\widehat{F}(t)$ in the brackets should not be confused with the
Helmholtz free energy $F(t)\equiv E(t)-T_{0}S(t)$. The above equalities should
be contrasted with the traditional inequalities $d_{\text{e}}W(t)\leq-dF(t)$
or $\Delta_{\text{e}}W(t)\leq-\Delta F(t)$.

\section{Applications\label{Sect_Applications}}

We now consider two simple applications of the new approach.

\subsection{\label{Subsection_Free_Expansion}Free Expansion}

Consider the example of free expansion in a set up in which gas in one chamber
is separated from another empty chamber of identical volume $V^{0}$ by a
partition; the latter is held stationary by some mechanism. Both chambers form
an isolated system so that not only $d_{\text{e}}Q\equiv0$ but also that the
expansion occurs at \emph{constant} energy $E^{0}$. Therefore, $d_{\text{e}%
}W\equiv0$. Therefore, the traditional formulation of the first law is of no
use in obtaining any useful information about the nature of irreversibility.
The initial pressure and temperature of the gas are $P^{0}$ and $T^{0}$. As
soon as the mechanism is removed, the partition becomes free to move and the
gas expands into the other chamber against a zero pressure. The pressure
difference $\Delta P(t)=P(t)$, where $P(t)$ is the instantaneous pressure of
the expanding gas occupying a volume $V(t)$; the latter can be recorded as a
function of time. Using the recorded $V(t)$, the temperature $T(t)$ and $P(t)$
of the gas are then obtained from the instantaneous equation of state of the
gas and the condition $E\left(  t\right)  =E^{0}$. The resulting irreversible
work is $d_{\text{i}}W(t)=P(t)dV(t)\equiv d_{\text{i}}Q(t)=T(t)d_{\text{i}%
}S(t)$, where we have used Eq. (\ref{Irreversible_Heat_02}) with $d_{\text{e}%
}S(t)=0$. We thus find%

\begin{equation}
\Delta_{\text{i}}S(t)=\int_{V^{0}}^{V(t)}P(t)dV(t)/T(t)=\int_{S^{0}}%
^{S(t)}\left.  dS\right\vert _{E}(t)=\left.  S\right\vert _{E}(t)-\left.
S\right\vert _{E}^{0},\nonumber
\end{equation}
where the suffix $E$ is used to represent that $E$ is constant in the process,
and where $\left.  S\right\vert _{E}^{0}$\ represents the initial entropy.
Neither the initial state nor the state at $t$ is required to be an
equilibrium state in the above calculation, which follows the recorded
non-equilibrium path given by $V(t)$ at constant $E$.

Assuming the initial and final states to be equilibrium states, we have
\[
\Delta_{\text{i}}S(t\rightarrow\infty)=S_{\text{eq}}(E^{0},2V^{0}%
,N)-S_{\text{eq}}(E^{0},V^{0},N).
\]
For an ideal gas, we evidently obtain $\Delta_{\text{i}}S(t\rightarrow
\infty)=N\ln2$ between the two equilibrium states, as expected.

The work done by and the heat generated within the system are given by%
\[
\Delta W(t)=\Delta_{\text{i}}Q(t)=\Delta_{\text{i}}W(t)=\int_{V^{0}}%
^{V(t)}P(t)dV(t),
\]
and can be calculated for the expansion profile given by $V(t)$.

\subsection{Relative Motion and the Resulting
Friction\label{Subsecton_Relative_Motion}}

Let us consider a different situation, in which the system is also moving with
respect to the medium with some velocity $\mathbf{V}(t)$ given by
\cite{Landau,Gujrati-II}
\[
\beta(t)\mathbf{V}(t)=-\left(  \partial S(t)/\partial\mathbf{P}(t)\right)
_{E(t),V(t)}%
\]
in which $\mathbf{P}(t)$ and $E(t)$ are the momentum and internal energy of
the system in the lab frame (in which $\Sigma_{0}$, and to a very good
approximation $\widetilde{\Sigma}$ are at rest). Such a motion arises during
sudden mixing of fluids or in a Couette flow, and eventually stops due to
deceleration caused by generated frictional forces as a result of the relative
motion; we do not consider any external force causing this motion. The
irreversible work done by the system against friction is%
\[
d_{\text{i}}^{\text{fr}}W(t)=-\mathbf{V}(t)\text{\textperiodcentered
}\mathbf{F}_{\text{fr}}dt>0,
\]
where $\mathbf{F}_{\text{fr}}$is the resulting frictional force opposing the
motion and results in $d\mathbf{P}(t)\equiv\mathbf{F}_{\text{fr}}dt$. It is
clear that $d_{\text{i}}^{\text{fr}}W(t)$ vanishes as the motion ceases. This
dissipative work is in addition to the irreversible work caused by any
pressure difference and any affinity as before:%
\begin{equation}
d_{\text{i}}W(t)=(P(t)-P_{0})dV(t)+A(t)d\xi(t)-\mathbf{V}%
(t)\text{\textperiodcentered}d\mathbf{P}(t).
\label{Irreversible_Friction_Work}%
\end{equation}
From Eq. (\ref{Irreversible_Heat_0}) and Eq.
(\ref{Irreversible_Heat_Work_equality}), we then find
\begin{equation}
T(t)d_{\text{i}}S(t)=[(P(t)-P_{0})dV(t)+A(t)d\xi(t)-\mathbf{V}%
(t)(t)\text{\textperiodcentered}d\mathbf{P}(t)-(T(t)-T_{0})d_{\text{e}%
}S(t).\nonumber
\end{equation}
\qquad

The same conclusion can also be obtained by considering the Gibbs fundamental
relation for this case. Using the above definition of the velocity, it is easy
to see that
\begin{equation}
dE(t)=T(t)dS(t)+\mathbf{V}(t)\text{\textperiodcentered}d\mathbf{P}%
(t)-P(t)dV(t)-A(t)d\xi(t); \label{Friction}%
\end{equation}
see Ref. \cite{Gujrati-II}. The last three terms, being isentropic
contribution to $dE(t)$\ represent $dW(t)$, while the first term represents
$dQ(t)$. We now follow Ref. \cite{Gujrati-II} and rewrite%
\begin{equation}
dW(t)=P_{0}dV(t)-\mathbf{V}_{0}\text{\textperiodcentered}d\mathbf{P}%
(t)+(P(t)-P_{0})dV(t)-(\mathbf{V}(t)-\mathbf{V}_{0})\text{\textperiodcentered
}d\mathbf{P}(t), \label{General_Friction_Form}%
\end{equation}
with the first two terms representing $d_{\text{e}}W(t)$, and the last two
terms $d_{\text{i}}W(t)$. Here, $\mathbf{V}_{0}$ represents the equilibrium
pressure velocity. However, the equilibrium value of $\mathbf{V}(t)$ is
$\mathbf{V}_{0}=0$, which then yields $d_{\text{i}}W(t)$ given in Eq.
(\ref{Irreversible_Friction_Work}). The exchange work $d_{\text{e}}%
W(t)=P_{0}dV(t)$ remains unchanged and says nothing about the presence of
internal friction terms that appears in our approach.

\section{Inclusion of other state variables\label{Sect_Other_State_Variables}}

Let us now extend the discussion to include other extensive quantities such as
the flow of matter, the electric interactions, chemical reactions, etc. For
specificity, we focus on chemical reactions, which we assume to be described
by a single extent of reaction $\xi(t)$. The corresponding affinity for the
system is given by $A(t)$, while that for the medium is given by $A_{0}=0$. We
assume another observable $X$ such as the number of solvent in a binary
mixture. The corresponding chemical potential is $\mu(t)$ for the system and
$\mu_{0}$\ for the medium. The work is now
\begin{equation}
dW(t)=P(t)dV(t)-\mu(t)dX+A(t)d\xi(t). \label{General_work}%
\end{equation}
The Gibbs fundamental relation for $\Sigma$ is given by
\begin{equation}
dE(t)=T(t)dS(t)-P(t)dV(t)+\mu(t)dX-A(t)d\xi(t), \label{General-Energy}%
\end{equation}
while the first law for it takes the form $dE(t)=dQ(t)-dW(t)$. Rewriting
$dQ(t)=dE(t)+P(t)dV(t)-\mu(t)dX+A(t)d\xi(t)$ as $dQ(t)=dE(t)+P_{0}%
dV(t)-\mu_{0}dX+(P(t)-P_{0})dV(t)-(\mu(t)-\mu_{0})dX+A(t)d\xi(t),$ we can
identify \cite{Gujrati-II} $d_{\text{e}}Q(t)$ with the first three terms%
\[
d_{\text{e}}Q(t)=dE(t)+P_{0}dV(t)-\mu_{0}dX.
\]
The last two terms represent $d_{\text{e}}W(t)$:%
\[
d_{\text{e}}W(t)=P_{0}dV(t)-\mu_{0}dX,
\]
which shows that the second term will appear in the traditional formulation of
the first law, but not the associated thermodynamic force $\mu_{0}-\mu(t)$;
see below. Similarly, $d_{\text{i}}Q(t)(t)=d_{\text{i}}W(t)$ represents the
dissipation given by%
\begin{equation}
d_{\text{i}}W(t)=(P(t)-P_{0})dV(t)-(\mu(t)-\mu_{0})dX+A(t)d\xi(t);
\label{Irreversible-Work_Extended}%
\end{equation}
see Eq. (\ref{Irreversible_Heat_Work_equality})$.$ According to the second
law, each contribution in $d_{\text{i}}W(t)$ is non-negative. The irreversible
entropy generation is given by%
\begin{equation}
T_{0}d_{\text{i}}S(t)=[T_{0}-T(t)]dS(t)+[P(t)-P_{0}]dV(t)+A(t)d\xi(t)+(\mu
_{0}-\mu(t))dX(t). \label{Irreversible-Entropy_Extended}%
\end{equation}
The extension to arbitrary number of observables and internal variables is
trivial \cite{Gujrati-II}.

\section{A Closed System Performing Work on an External
Object\label{Sec-Closed-System}}

We now consider our isolated system to consist of the previous medium
$\widetilde{\Sigma}$($T_{0},P_{0},A_{0})$, the system $\Sigma\lbrack
T(t),P(t),A(t),\mu(t)]$ and an external object $\Sigma_{\text{ext}}%
[T_{0},P_{0},A_{0},\mu_{\text{ext}}(t)];\ $the latter is completely
disconnected from the medium but can only exchange the type of work with the
system involving the extensive variable $X(t)$ considered in the previous
section. However, there are some important differences. First, we do not limit
the external object to be extremely large compared to the system. Accordingly,
its chemical potential $\mu_{\text{ext}}(t)$ can change in time. Just for
convenience, we assign to it $T_{0},P_{0},A_{0}$ of the medium, even though it
is isolated from the latter and cannot exchange heat and volume and internal
variable work with the system and the medium. It is a classic problem that has
been extensively studied in which the external object is thermally insulated
form the system and the medium; see for example Ref. \cite[see Sec. 20
there]{Landau}. It is also studied by the stochastic trajectory approach; see
for example, Refs. \cite{Jarzynski,Seifert}. Here, we will closely follow
Landau and Lifshitz \cite[see Sec. 20 there]{Landau}. As the external object
only performs work, its entropy $S_{\text{ext}}$ must not change. We will
denote the totality of $\Sigma$\ and $\Sigma_{\text{ext}}\ $by $\overline
{\Sigma}$ and all its quantities by an additional bar. For the isolated system
$\overline{\Sigma}$, we have $\Delta\overline{S}_{0}(t)\equiv\Delta_{\text{i}%
}\overline{S}_{0}(t)=\Delta\left[  S(t)+S_{\text{ext}}(t)+\widetilde
{S}(t)\right]  =\Delta\left[  S(t)+\widetilde{S}(t)\right]  \geq0$ during a
process involving (not necessarily small) change within the isolated system
$\overline{\Sigma}_{0}$. As usual, all exchange quantities between various
pairs of the three parts cancel out for $\overline{\Sigma}_{0}$; only
irreversible quantities survive such as $\Delta_{\text{i}}\overline{S}_{0}%
(t)$, which in our approach is associated with the entropy generation within
$\overline{\Sigma}$. Accordingly, we will use the notation $\Delta_{\text{i}%
}\overline{S}(t)$ for $\Delta_{\text{i}}\overline{S}_{0}(t)$ below and write
\begin{equation}
\Delta_{\text{i}}\overline{S}(t)=\Delta_{\text{i}}S^{(\text{Q})}%
(t)+\Delta_{\text{i}}S^{(\text{V})}(t)+\Delta_{\text{i}}S^{(\xi)}%
(t)+\Delta_{\text{i}}\overline{S}^{(\text{X})}(t)\geq0.
\label{Irreversible_entropy_Closed}%
\end{equation}
The first three terms on the right in the first equation denote the cumulative
irreversible entropy changes due to exchanges between the system and the
medium and the last term is the cumulative irreversible entropy change due to
the exchange between the system and the external object. The work done
$\Delta_{\text{e}}R(t)$ by the external object on the system is%
\begin{equation}
\Delta_{\text{e}}R(t)=\Delta\left[  E(t)-T_{0}S(t)+P_{0}V(t)\right]
+T_{0}\Delta_{\text{i}}\overline{S}(t). \label{Work_Closed-0}%
\end{equation}
In terms of $\Delta_{\text{e}}W(t)=-\Delta_{\text{e}}R(t)$, we have the
following inequality%
\begin{equation}
\Delta_{\text{e}}W(t)\leq\Delta G(T_{0},P_{0},t)\equiv\Delta\left[
E(t)-T_{0}S(t)+P_{0}V(t)\right]  \label{Work_Closed-System}%
\end{equation}
where $\Delta G(T_{0},P_{0},t)$ is the change in the Gibbs free energy of the
system. It should be stressed that the above inequality does not assume the
establishment of internal equilibrium in the system. If we are considering
spontaneous changes occurring within the system in the medium $\widetilde
{\Sigma}$($T_{0},P_{0},A_{0})$, we must set the external work $d_{\text{e}%
}R=0$ above by setting $dX(t)=0$. We then find%
\begin{equation}
dG(T_{0},P_{0},t)=-T_{0}d_{\text{i}}S(t)\leq0
\label{Gibbs_-Free-Energy-Behavior}%
\end{equation}
during spontaneous relaxation. Thus, assigning the irreversibility to the
system leads to the expected behavior of the time-dependent Gibbs free energy
of the system.

The corresponding work done by the system with respect to the variable $X(t)$
is $dW_{\text{X}}(t)=-\mu(t)dX(t)$, while the work done by the external object
is $d_{\text{e}}R(t)=$ $\mu_{\text{ext}}dX(t)$, with the result $d_{\text{i}%
}W_{\text{X}}(t)=\left[  \mu_{\text{ext}}(t)-\mu(t)\right]  \geq0$. We then
find%
\begin{equation}
T_{0}d_{\text{i}}\overline{S}(t)=[T_{0}-T(t)]dS(t)+[P(t)-P_{0}]dV(t)+A(t)d\xi
(t)+\left[  \mu_{\text{ext}}(t)-\mu(t)\right]  dX(t)\geq0,
\label{Irreversible_entropy_Closed-0}%
\end{equation}
in which the first three terms refer to irreversible entropy generation within
the system and the last term refers to the generation within $\overline
{\Sigma}$. Each term on the right must be non-negative; compare with Eq.
(\ref{Irreversible-entropy-contributions}). Comparison with Eq.
(\ref{Irreversible-Work_Extended}) \ shows that the last three terms above
represent $d_{\text{i}}W(t)$, except that $\mu_{0}$ is replaced by
$\mu_{\text{ext}}$. The additional thermodynamic force $\mu_{\text{ext}%
}(t)-\mu(t)$ vanishes when there is equilibrium established between the system
and the external object; let the equilibrium value be given by $\overline{\mu
}_{0}$. The new force generates the irreversible entropy $d_{\text{i}%
}\overline{S}^{\text{(X)}}(t)$ due to $X$. Following the arguments given in
Ref. \cite{Gujrati-II}, see Sec. XI\textbf{ }there, for the case of thermal
equilibration and which can be adapted here verbatim, we find that the
irreversible entropies for $\Sigma$\ and $\Sigma_{\text{ext}}\ $are given,
respectively, by
\begin{equation}
T_{0}d_{\text{i}}S^{(\text{X)}}(t)=\left[  \overline{\mu}_{0}-\mu(t)\right]
dX(t),T_{0}d_{\text{i}}S_{\text{ext}}^{\text{(X)}}(t)=\left[  \overline{\mu
}_{0}-\mu_{\text{ext}}(t)\right]  dX_{\text{ext}}(t), \label{Irreversible_S_X}%
\end{equation}
with $dX(t)+dX_{\text{ext}}(t)=0$. As the entropy of the external object
remains constant, there must be a corresponding change $d_{\text{e}%
}S_{\text{ext}}^{\text{(X)}}(t)=-d_{\text{i}}S_{\text{ext}}^{\text{(X)}}(t)$
as it performs work, which requires $d_{\text{e}}p_{k\text{,ext}%
}(t)=-d_{\text{i}}p_{k\text{,ext}}(t)$ for the $k$-th microstate of the object.

If the object is extremely large compared to the system so that the coupling
with the system will not appreciably affect $\mu_{\text{ext}}$, the latter can
be taken as a constant equal to $\mu_{0}=\overline{\mu}_{0}$. In this
situation, the entire contribution $d_{\text{i}}\overline{S}^{\text{(X)}}(t)$
must be treated as $d_{\text{i}}S^{\text{(X)}}(t)$\ assigned to the system, as
there is no force in the object to generate irreversibility ($d_{\text{i}%
}S_{\text{ext}}(t)=0$).

The difference between our approach and the stochastic trajectory approach
should be pointed out. In the latter approach, $T=T_{0}$ and $P=P_{0}$ for the
system. Also, as no internal variable is considered, we will set $A=A_{0}=0$
above. Thus, the only irreversibility is due to the possibility of work by the
external object. The entropy $S(t)$ is given by the Gibbs formulation; see Eq.
(\ref{Entropy_Energy}), and its changes $dS(t)$ includes the irreversible
entropy changes $d_{\text{i}}S(t)=d_{\text{i}}S^{(\text{X)}}(t)$. This is not
the case in the stochastic trajectory approach in which the irreversible
entropy $\Delta_{\text{i}}\overline{S}^{(\text{X})}(t)$ is assigned to the
medium \cite{note-2}.

\section{Summary and Discussion}

\subsection{Summary}

The traditional formulation of the first law has two basic deficiencies. The
first one is that even though the first law is about the change in the energy
of the system, it contains the fields of the medium and not of the system. The
second, and more dramatic conceptually, deficiency is that it remains
satisfied even if the process violated the second law. This is why both laws
are needed in nonequilibrium thermodynamics. Despite the concern expressed by
Kestin \cite[Sect. 5.12]{Kestin}, we have shown that the definition of the
generalized heat $dQ(t)$, which includes its irreversible component
$d_{\text{i}}Q(t)$, follows \emph{uniquely} from the unique choice of $dW(t)$
resulting from the second law, see Eq. (\ref{Work_Forms}), and also from the
statistical approach, see Eq. (\ref{dE_dW}). It general, we have the identity
$d_{\text{i}}W(t)\equiv$ $d_{\text{i}}Q(t)$ between their irreversible
components; see Eq. (\ref{Irreversible_Heat_Work_equality}). These results do
not require the system to be in internal equilibrium. The situation becomes
very illuminating when the system is in internal equilibrium. In this case,
the generalized heat formulation gives rise to not only the simple
relationship $dQ(t)=T(t)dS(t)$, but also the extension of the adiabatic
theorem for nonequilibrium states; see Theorem \ref{Adiabatic Theorem}.
Another remarkable consequence is that in terms of generalized $dQ(t)$,\ the
Clausius equality (\ref{Clausius_Equality_0}) is always maintained, in
contrast to the inequality (\ref{Clausius_Inequality}) in the traditional
approach. Our generalized formulation brings about a symmetry between the
system and its surrounding medium, see Eq. (\ref{Work_Forms}), which is absent
in the traditional approach using external fields. The first law becomes
identical to the second law so that we only deal with equalities and a single
law; the equalities are easier to deal with than the inequalities that result
in the traditional approach; see for example, Eq.
(\ref{Work_FreeEnergy_Equality}) or Eq. (\ref{Gibbs_Fundamental_Relations}).

The introduction of microstate probabilities and the partition of $dp_{i}(t)$
into $d_{\text{e}}p_{i}(t)$ and $d_{\text{i}}p_{i}(t)$ provide a very
convenient and useful statistical representation of exchange and irreversible
forms of heat and entropy and of other quantities, and which allows us to draw
some useful conclusions. For example, it follows directly from Theorem
(\ref{Adiabatic Theorem}) that pure work cannot change the entropy. This was
used to argue that $S_{\text{ext}}$\ must be a constant in Sec.
\ref{Sec-Closed-System}. We have found that the microstate probability
$p_{i}(t)$ is very different from its equilibrium analog $p_{i\text{,eq}}.$
Several examples have been given to show the usefulness of the new definition
of heat and work. Below we list some of the benefits of the new formulation.

\subsection{Benefits of the Generalized Formulation\label{Sect_Benefits}}

\begin{enumerate}
\item It is a new way to write the first law, which shows that the first and
the second laws are \emph{no longer} independent, a property now shared by all
irreversible processes and not just by reversible processes. This conclusion
must be contrasted with the well-known result found in all text books and the
literature that the two laws are independent for irreversible processes.

\item Work and heat can be uniquely distinguished so that each has a clear
physical significance, but no such clear identification is possible in the
traditional formulation of the first law; see Theorems
(\ref{Theorem_Heat_Work}) and (\ref{Theorem_Traditional_Heat_Work}).

\item Work and heat no longer are restricted to flow through the boundary.
Thus, the new formulation can easily accommodate interactions described by
fields (electromagnetic, gravitational, etc.) \cite{Kestin}.

\item The entire thermodynamics and stability considerations for any body can
be expressed in terms of the variables associated with the body alone at each
instant. We can use the instantaneous equations of state of the body alone for
thermodynamic computation. Indeed, it may be the only way to study a
non-equilibrium system in some cases as is clearly shown by the example of
free expansion in Sect. \ref{Sect_Applications}. Here, the traditional
formulation of the first law is of no use to determine $d_{\text{i}}W(t)$ and
$d_{\text{i}}S(t)$\ as is well known.

\item We deal with equalities and not inequalities. In addition, we only deal
with $dS(t)$, which as a dependent variable is expressed in terms of state
variables $\mathbf{Z}(t)$; there is no need to consider its parts
$d_{\text{e}}S(t)$ and \ \ $d_{\text{i}}S(t)$ separately, which depend on the
medium; see (3) in Sect. \ref{Subsection_Limitations}. The generalized heat
$dQ(t)$\ and work $dW(t)$ differ from $d_{\text{e}}Q(t)$ and $d_{\text{e}%
}W(t)$, respectively, by the same contribution $d_{\text{i}}Q(t)\equiv
d_{\text{i}}W(t)\equiv d_{\text{i}}W_{0}(t)$.

\item The use of the Gibbs fundamental relation makes it almost trivial to
identify the irreversible contributions. Thus, the new formulation does not
have the limitation inherent in the traditional formulation; see (7) in Sect.
\ref{Subsection_Limitations}. The determination of $d_{\text{i}}W_{0}(t)$ is
straightforward by measuring $P(t)$ and $P_{0}$. Adding this to $d_{\text{e}%
}Q(t)$ then allows us to determine $dQ(t)$.

\item The most useful aspect of the generalized approach is that the work
$dW(t)$ does not depend on the amount and nature of dissipation going on
during a process. It is determined by the instantaneous pressure of the
system, which itself is a derivative of its instantaneous entropy, a state
variable. We have already seen this earlier in Eq. (\ref{Dissipation_work}).
Whether friction between the piston and the walls is present or not,
$d_{\text{i}}W(t)$ is always given by $\Delta P(t)dV(t)$ so that $dW(t)=$
$P(t)dV(t)$.

\item Once $dQ(t)$ is known, $dS(t)\equiv dQ(t)/T(t)$ is also known or vice
versa as if we are dealing with a system in equilibrium, even though we are
dealing with an out-of-equilibrium system in internal equilibrium. In the
latter situation, the analysis is considerably simplified as we are always
dealing with the instantaneous properties such as the equation of state of the
system even though the system is out of equilibrium. We do not need to, or may
not, know these properties for $\widetilde{\Sigma}$.
\end{enumerate}

The new formulation has many other desirable properties. The generalized heat
$dQ(t)$\ and work $dW(t)$ differ from $d_{\text{e}}Q(t)$ and $d_{\text{e}%
}W(t)$, respectively, by the same contribution $d_{\text{i}}Q(t)\equiv
d_{\text{i}}W(t)$. The determination of $d_{\text{i}}W(t)$ is straight forward
in terms of thermodynamic forces $P(t)-P_{0},A(t)$, etc. Adding this to
$d_{\text{e}}Q(t)$ then allows us to determine $dQ(t)$ from which $dS(t)$ can
be determined.


\begin{thebibliography}{99}
\bibitem {note-2}The temperature of the system and any possible variation in
it is never discussed in this approach, which makes the dynamics in this
approach presumably isothermal, as already pointed out in Ref. \cite{Cohen}.
It is also not discussed whether there is always a mechanical equilibrium
between the external diriving object and the system. If there is mechanical
equilibrium, there cannot be any irreversibility. If it is absent, there will
be irreversible entropy geneartion. However, in this approach, this entropy is
ascribed to the medium whose temperature is constant, which appears contrary
to the well-known fact that the heat transfer at constant temperature is
always reversible. In addition, the work and heat in this approach refer to
exchange work and heat, even though the formulation of work is no different
from our isentropic formulation in Eq. (\ref{Work_system}); see for example
Refs.
\cite{Nieuwenhuizen,Jarzynski,Seifert,Jarzynski0,Crooks,Ritort,Alicki,Deffner}%
. \ \ \ \ \ \ \ \ \ \ \ \ \ \ \ \ \ \ \ \ \ \ \ \ \ \ \ \ \ \ \ \ \ \ \ \ \ \ \ \ \ \ \ \ \ \ \ \ \ \ \ \ \ \ \ \ \ \ \ \ \ \ \ \ \ \ \ \ \ \ \ \ \ \ \ \ \ \ \
\bibitem {note-3}This should be contrasted with what is done in the stochastic
trajectory approach; see \cite{note-2}.
\bibitem {note-4}We can classify the microstates as follows. We first fix $V$
and $\xi$ and index the microstates by $i(V,\xi)$. We do this for all possible
pairs of $V$ and $\xi$. Thus, the sum over the index $i$ is a shorthand
notation for a sum over $V,\xi$ and $i(V,\xi)$.
\bibitem {Eu}It should be noted that Eu [B.C. Eu, Phys. Rev. E \textbf{51},
768 (1995)] had introduced a compensation function $\Psi$ for which $%
%TCIMACRO{\toint }%
%BeginExpansion
{\textstyle\oint}
%EndExpansion
d\Psi=0$, but did not recognize that Clausius' $N(t,\tau)$ was given by Eq.
(\ref{Irreversible_entropy_Cycle}). \
\end{thebibliography}
\end{document}